\documentclass[a4paper,UKenglish,cleveref, autoref, thm-restate]{lipics-v2021}

\pdfoutput=1 
\hideLIPIcs  

\usepackage[utf8]{inputenc}

\usepackage{tabularx}

\usepackage{amsthm}
\usepackage{amsmath}
\usepackage{amssymb}
\usepackage{amsfonts}
\usepackage[sort,numbers]{natbib}
\usepackage{framed}
\usepackage{xcolor}

\usepackage{xspace}

\usepackage{todonotes}

\crefname{observation}{Observation}{Observations}
\Crefname{observation}{Observation}{Observations}

\DeclareMathOperator{\vimw}{vimw}
\DeclareMathOperator{\poly}{poly}
\DeclareMathOperator{\inc}{inc}
\DeclareMathOperator{\appears}{appears}
\DeclareMathOperator{\equal}{equal}
\DeclareMathOperator{\lessthan}{lessthan}

\DeclareMathOperator{\ispath}{ispath}
\DeclareMathOperator{\degone}{degone}
\DeclareMathOperator{\degtwo}{degtwo}
\DeclareMathOperator{\conn}{conn}
\DeclareMathOperator{\paths}{paths}
\DeclareMathOperator{\completions}{completions}
\DeclareMathOperator{\sample}{Sample}

\DeclareMathOperator{\sol}{Sol}

\newcommand\sigmageneric[1]{\sigma^{(\star)}_{#1}}

\newcommand{\Paths}{\textsc{\#Path}\xspace}
\newcommand{\TemporalPaths}{\textsc{\#Temporal Path}\xspace}
\newcommand{\STemporalPaths}{\textsc{\#Short Temporal Path}\xspace}
\newcommand{\MTemporalPaths}{\textsc{\#Multicoloured Temporal Path}\xspace}

\newcommand{\problemdef}[3]{
	\begin{center}\fbox{
	\begin{minipage}{0.95\textwidth}
		\noindent
		#1
		\vspace{5pt}\\
		\setlength{\tabcolsep}{3pt}
		\begin{tabularx}{\textwidth}{@{}lX@{}}
			\textrm{Input:}     & #2 \\
			\textrm{Task:}  & #3
		\end{tabularx}
	\end{minipage}}
	\end{center}
}

\newcommand{\commentout}[1]{}

\usepackage{tikz}

\usepackage{pgfplots}
\pgfplotsset{compat=1.10}
\usepgfplotslibrary{fillbetween}
\usetikzlibrary{intersections}

\usetikzlibrary{arrows,decorations.pathmorphing,decorations.pathreplacing,backgrounds,positioning,fit,matrix}
\usetikzlibrary{shapes,calc,patterns,arrows.meta}
\tikzset{
	vert/.style={circle,inner sep=1.5,fill=white,draw,minimum size=.3cm},
	edge/.style={color=black, thick},
	diredge/.style={->,>={Stealth[width=8pt,length=8pt]},color=black, thick},
	timelabel/.style={fill=white,font=\footnotesize, text centered}
}

\title{Counting Temporal Paths}

\author{Jessica Enright}{School of Computing Science, University of Glasgow, UK}{jessica.enright@glasgow.ac.uk}{}{Supported by EPSRC grant EP/T004878/1.}

\author{Kitty Meeks}{School of Computing Science, University of Glasgow, UK}{kitty.meeks@glasgow.ac.uk}{https://orcid.org/0000-0001-5299-3073}{Supported by EPSRC grants EP/T004878/1 and EP/V032305/1.}

\author{Hendrik~Molter}{Department of Computer Science and Department of Industrial Engineering and Management, Ben-Gurion~University~of~the~Negev, 
Beer-Sheva, 
Israel}{molterh@post.bgu.ac.il}{https://orcid.org/0000-0002-4590-798X}{Supported by the ISF, grants No.~1456/18 and No.~1070/20, and European Research Council, grant number 949707.}

\authorrunning{Jessica Enright, Kitty Meeks, and Hendrik Molter} 

\Copyright{Jessica Enright, Kitty Meeks, and Hendrik Molter} 



\category{} 

\relatedversion{} 

\funding{}

\acknowledgements{This work was initiated at the Dagstuhl Seminar ``Temporal Graphs: Structure, Algorithms, Applications'' (Dagstuhl Seminar Nr.~21171).}

\nolinenumbers 

\EventEditors{John Q. Open and Joan R. Access}
\EventNoEds{2}
\EventLongTitle{42nd Conference on Very Important Topics (CVIT 2016)}
\EventShortTitle{CVIT 2016}
\EventAcronym{CVIT}
\EventYear{2016}
\EventDate{December 24--27, 2016}
\EventLocation{Little Whinging, United Kingdom}
\EventLogo{}
\SeriesVolume{42}
\ArticleNo{23}

\begin{document}

\maketitle

\begin{abstract}
The betweenness centrality of a vertex $v$ is an important centrality measure that quantifies how many optimal paths between pairs of other vertices visit $v$.  Computing betweenness centrality in a temporal graph, in which the edge set may change over discrete timesteps, requires us to count temporal paths that are optimal with respect to some criterion.  For several natural notions of optimality, including \emph{foremost} or \emph{fastest} temporal paths, this counting problem reduces to \TemporalPaths, the problem of counting \emph{all} temporal paths between a fixed pair of vertices; like the problems of counting foremost and fastest temporal paths, \TemporalPaths is \#P-hard in general.  Motivated by the many applications of this intractable problem, we initiate a systematic study of the parameterised and approximation complexity of \TemporalPaths.  We show that the problem presumably does not admit an FPT-algorithm for the feedback vertex number of the static underlying graph, and that it is hard to approximate in general.  On the positive side, we prove several exact and approximate FPT-algorithms for special cases.
\end{abstract}

\section{Introduction}
Computing a (shortest) path between two vertices in a graph is one of the most important tasks in algorithmic graph theory and serves as a subroutine in a wide variety of algorithms for connectivity-related graph problems. 
The \emph{betweenness centrality} measure for vertices in a graph was introduced by \citet{freeman_set_1977} and motivates the task of 
 \emph{counting} shortest paths in a graph. Intuitively, betweenness centrality measures the importance of a vertex for information flow under the assumption that information travels along optimal (i.e.\ shortest) paths. More formally, the betweenness of a vertex~$v$ is based on the ratio of the number of shortest paths between vertex pairs that visit~$v$ as an intermediate vertex and the total number of shortest paths, thus its computation is closely related to shortest path counting.
The betweenness centrality is a commonly used tool in network analysis and it can be computed in polynomial time; e.g.\ Brandes' algorithm~\cite{brandes_faster_2001} serves as a blueprint for all modern betweenness computation algorithms and implicitly also counts shortest paths.

In contrast to the tractability of counting shortest paths, the problem of counting \emph{all} paths between two vertices in a graph is one of the classic problems discussed in the seminal paper by~\citet{valiant_complexity_1979} that is complete for the complexity class \#P (the counting analogue of NP) and hence is presumably not doable in polynomial time.

\emph{Temporal} graphs are a natural generalisation of graphs that capture dynamic changes over time in the edge set. They have a fixed vertex set and a set of \emph{time-edges} which have integer time labels indicating at which time(s) they are active.
In recent years, the research field of studying algorithmic problems on temporal graphs has steadily grown~\cite{Hol15,HS19,LVM18,Mic16}. In particular, an additional layer of complexity is added to connectivity related problems in the temporal setting. Paths in temporal graphs have to respect time, that is, a \emph{temporal path} has to traverse time-edges with non-decreasing time labels~\cite{KKK02}\footnote{Temporal paths that traverse time-edge with non-decreasing time labels are often referred to as ``non-strict'', in contrast to \emph{strict} temporal paths, which traverse time-edges with increasing time labels. In this work, we focus on non-strict temporal paths.}. This implies that temporal connectivity is generally not symmetric and not transitive, a major difference from the non-temporal case. Furthermore, there are several natural optimality concepts for temporal paths, the most important being \emph{shortest}, \emph{foremost}, and \emph{fastest} temporal paths~\cite{xuan_computing_2003}. Intuitively speaking, shortest temporal paths use a minimum number of time-edges, foremost temporal paths arrive as early as possible, and fastest temporal paths have a minimum difference between start and arrival times. We remark that an optimal path with respect to any of these three criteria can be found in polynomial time~\cite{xuan_computing_2003,wu_efficient_2016}.
The existence of multiple natural optimality concepts for temporal paths implies several natural definitions of temporal betweenness, one for each path optimality concept~\cite{RymarMNN21,BMNR20,LVM18}.

Similar to the non-temporal case, the ability to \emph{count} optimal temporal paths is a key ingredient for the corresponding temporal betweenness computation. However, the picture is more complex in the temporal setting. 
Shortest temporal paths can be counted in polynomial time and the corresponding temporal betweenness can be computed in polynomial time~\cite{BMNR20,kim_temporal_2012,habiba2007betweenness,RymarMNN21}.
In contrast, counting foremost or fastest temporal paths is \#P-hard~\cite{rad2017computation,BMNR20,MO18}, which implies that computing the corresponding temporal betweenness is \#P-hard as well~\cite{BMNR20}. Indeed, \citet{BMNR20} show that there is a polynomial time reduction from the problem of counting foremost or fastest temporal paths to the problem of the corresponding temporal betweenness computation. Note that a reduction in the other direction is straightforward.

In this work, we study the (parameterised) computational complexity of (approximately) counting foremost or fastest temporal paths. In fact, we study the simpler and arguably more natural problem of counting all temporal paths from a start vertex $s$ to a destination vertex $z$ in a temporal graph. 

Let $\mathcal{G}=(V,\mathcal{E},T)$ denote a temporal graph with vertex set $V$, time-edge set $\mathcal{E}$, and maximum time label (or lifetime)~$T$ (formal definitions are given in \cref{sec:prelims}).  We are then concerned with the following computational problem:

\problemdef{\TemporalPaths}{A temporal graph $\mathcal{G}=(V,\mathcal{E},T)$ and two vertices $s,z\in V$.}{Count the temporal $(s,z)$-paths in $\mathcal{G}$.}

It is easy to see that \TemporalPaths generalises the problem of counting paths in a non-temporal graph (all time-edges have the same time label), hence we deduce that \TemporalPaths is \#P-hard.
Furthermore, observe that using an algorithm for \TemporalPaths, it is possible to count foremost or fastest temporal paths with only polynomial overhead in the running time; we discuss this reduction in more detail in \cref{sec:countvsbetweenness}. Hence, all exact algorithms we develop for \TemporalPaths can be used to compute the temporal betweenness based on foremost or fastest temporal paths with polynomial overhead in the running time. To the best of our knowledge, this is the first attempt to systematically study the parameterised complexity and approximability of \TemporalPaths.

\subsection{Related Work}
As discussed above, the temporal setting adds a new dimension to connectivity-related problems. The problems of computing shortest, foremost, and fastest temporal paths have been studied thoroughly~\cite{xuan_computing_2003,wu_efficient_2016,BHNN20}. The temporal setting also offers room for new natural temporal path variants that do not have an analogue in the non-temporal setting. \citet{CHMZ21} study the problem of finding \emph{restless temporal paths} that dwell an upper-bounded number of time steps in each vertex, while \citet{FMNR22} study the problem of finding  \emph{delay-robust routes} in a temporal graph (intuitively, temporal paths that are robust with respect to edge delays); both problems turn out to be NP-hard.

The problem of \emph{counting} (optimal) temporal paths has mostly been studied indirectly in the context of temporal betweenness computation.
However, the computation of temporal betweenness has received much attention~\cite{BMNR20,rad2017computation,RymarMNN21,kim_temporal_2012,tang_analysing_2010,habiba2007betweenness,tang2013applications,alsayed2015betweenness,TBBLS20,nicosia_graph_2013,SML21}. Most of the mentioned work considers temporal betweenness variants that are polynomial-time computable. The corresponding optimal temporal paths are mostly shortest temporal paths or variations thereof. There are at least three notable exceptions: \citet{BMNR20} also consider \emph{prefix-foremost} temporal paths and the corresponding temporal betweenness and show that the latter is computable in polynomial time. Furthermore they show \#P-hardness for several temporal betweenness variants based on \emph{strict} optimal temporal paths. \citet{rad2017computation} consider temporal betweenness based on foremost temporal paths and show that its computation is \#P-hard. They further give an FPT-algorithm to compute temporal betweenness based on foremost temporal paths for the number of vertices as a parameter (note that the size of a temporal graph generally cannot be bounded by a function of the number of its vertices). 
\citet{RymarMNN21} give a quite general sufficient condition called \emph{prefix-compatibility} for optimality concepts for temporal paths that makes it possible to compute the corresponding temporal betweenness in polynomial time. 

Generally, connectivity related problems have received a lot of attention in the temporal setting, ranging from the mentioned temporal path and betweenness computation to finding temporally connected subgraphs~\cite{AF16,CasteigtsPS21}, temporal separation~\cite{Flu+20,KKK02,MaackMNR21,Zsc+20,Molter22}, temporal graph modification to decrease or increase its connectivity~\cite{DeligkasP20,EnrightMMZ21,MolterRZ21,enright2021assigning}, temporal graph exploration~\cite{Erlebach0K21,AkridaMSR21,ErlebachKLSS19,ErlebachS20,BodlaenderZ19,BumpusM21}, temporal network design~\cite{KlobasMMS22,MertziosMS19,Akr+17}, and others~\cite{HaagMNR22,FMNR222,KlobasMMNZ21}.

In the static setting the general problem of counting $(s,z)$-paths in static graphs is known to be \#P-complete~\cite{valiant_complexity_1979}.  In the parameterised setting, the problem of counting length-$k$ paths (with parameter $k$) was one of the first problems shown to be \#W[1]-complete \cite{flum2004parameterized}, but the problem does admit an efficient parameterised approximation algorithm \cite{arvindraman02}.  It is also generally considered folklore that the problem of counting paths (of any length) admits an FPT-algorithm parameterised by the treewidth of the input graph.

\subsection{Our Contribution}
Our goal is to initiate the systematic study of the parameterised and approximation complexity of \TemporalPaths.  
We provide an argument that \TemporalPaths is essentially equivalent to counting foremost or fastest paths or computing the respective temporal betweenness centrality in \cref{sec:countvsbetweenness}.

\subparagraph{Hardness results (\cref{sec:intractability}).} The main technical contribution of this paper is a reduction showing that \TemporalPaths is intractable even when very strong restrictions are placed on the underlying graph; specifically the problem is hard for $\oplus$W[1] when parameterised by the feedback vertex number of the underlying graph, which rules out the existence of FPT algorithms with respect to several common parameters.  We also show that it is NP-hard even to approximate the number of temporal $(s,z)$-paths in general, motivating the study of approximate counting in more restricted settings.

\subparagraph{Exact algorithms for special cases (\cref{sec:exactalgs}).}  We show that the problem is polynomial-time solvable if the underlying graph is a forest, and then use a wide range of algorithmic techniques to generalise this result in different ways. We show that the problem is
fixed-parameter tractable with respect to two ``distance to forest'' parameterisations that are larger than the feedback vertex number of the underlying graph (timed feedback vertex number and underlying feedback edge number).  We further show that \TemporalPaths is in FPT parameterised by the treewidth of the underlying graph and the lifetime combined, or parameterised by the recently introduced parameter ``vertex-interval-membership-width''.

\subparagraph{Approximation algorithms (\cref{sec:approx}).}  We show that there is an FPTRAS for \TemporalPaths parameterised by the maximum permitted length of a temporal $(s,z)$-path.  We then turn our attention to the problem of approximating betweenness centrality, as the relationship between path counting and computing betweenness is not so straightforward in the approximate setting: we demonstrate that, whenever there exists an FPRAS (respectively FPTRAS) for \TemporalPaths, we can efficiently approximate the maximum betweenness centrality of any vertex in the temporal graph.  These two results together give an FPTRAS to estimate the maximum betweenness centrality of any vertex in a temporal graph (with respect to either foremost or fastest temporal paths) parameterised by the vertex cover number or treedepth of the underlying input graph.

\section{Preliminaries and Basic Observations}\label{sec:prelims}
In this section we provide all basic notations, definitions, and terminology used in this work.
We discuss the relation between temporal path counting and temporal betweenness computation in more detail in \cref{sec:countvsbetweenness}.
Additional background on parameterised and approximate counting complexity are given in \cref{sec:paramcompl,sec:approxcompl}, respectively.
Given a static graph $G=(V,E)$, we say that a sequence $P=\left(\{v_{i-1},v_i\}\right)_{i=1}^k$ of edges in $E$ forms a \emph{path} in $G$ if $v_{i}\neq v_j$ for all $0\le i<j\le k$.

\subsection{Temporal Graphs and Paths}
There are several different definitions and notations used in the context of temporal graphs~\cite{Hol15,HS19,LVM18,Mic16} which are mostly equivalent. Here, we use the following definitions and notations:

An (undirected, simple) \emph{temporal graph} with lifetime $T\in\mathbb{N}$ is a tuple~$\mathcal{G}=(V,\mathcal{E},T)$, with time-edge set $\mathcal{E}\subseteq\binom{V}{2}\times [T]$.
We assume all temporal graphs in this paper to be undirected and
simple.
The \emph{underlying graph} of $\mathcal{G}$ is defined as the static graph~$G=(V, \{\{u,v\}\mid \exists t\in [T] \text{ s.t.\ } (\{u,v\},t)\in\mathcal{E}\})$.
We denote by $E_t$ the set of edges of $G$ that are active at time $t$, that is, $E_t=\{\{u,v\}\mid (\{u,v\},t)\in\mathcal{E}\}$.

For every $v\in V$ and every time step $t\in [T]$, we denote the \emph{appearance
of vertex} $v$ \emph{at time}~$t$ by the pair $(v,t)$. 
For a time-edge $(\{v,w\},t)$ we call the vertex appearances $(v,t)$ and $(w,t)$ its \emph{endpoints} and we call $\{v,w\}$ its \emph{underlying edge}.

We assume that every number in $[T]$ appears at least once as a label for an edge in $\mathcal{E}$. In other words, we ignore labels that are not used for any edges since they are irrelevant for the problems we consider in this work. It follows that we assume $T\le |\mathcal{E}|$ and hence $T\in \mathcal{O}(|\mathcal{G}|)=\mathcal{O}(|V|+|\mathcal{E}|)$.

A \emph{temporal ($s,z$)-path} (or \emph{temporal path}) of length~$k$ from vertex $s=v_0$ to vertex $z=v_k$ in a temporal graph~$\mathcal{G}=(V,\mathcal{E},T)$ is a sequence 
$P = \left((\{v_{i-1},v_i\},t_i)\right)_{i=1}^k$ of time-edges in $\mathcal{E}$
such that the corresponding sequence of underlying edges forms a path in the underlying graph of $\mathcal{G}$ and, for
all $i\in [k-1]$, we have that $t_i \leq t_{i+1}$.
Given a temporal path $P = \left((\{v_{i-1},v_i\},t_i)\right)_{i=1}^k$, we denote the set of vertices of $P$ by $V(P)=\{v_0,v_1,\ldots,v_k\}$ and we say that $P$ \emph{visits} the vertex $v_i$ if $v_i\in V(P)$.
Moreover, we call vertex appearances $(v_{i-1},t_i)$ \emph{outgoing} for $P$ and we call the vertex appearances $(v_i,t_i)$ \emph{incoming} for $P$. Note that, if $t_i = t_{i+1}$, then $(v_i,t_i)$ is both incoming and outgoing for $P$.  We define $(v_0,1)$ to be incoming for $P$ and $(v_k,T)$ to be outgoing for $P$. We say that a vertex appearance is \emph{visited} by $P$ if it is outgoing or incoming for $P$ (so a vertex is visited by $P$ if and only if at least one of its appearances is visited by~$P$). We say that $P$ \emph{starts} at $v_0$ at time $t_1$ and \emph{arrives} at $v_k$ at time $t_k$.
We say that $P'$ is a \emph{temporal subpath} of $P$ if $P'$ is a subsequence of $P$.
Furthermore, we define the following optimality concepts for temporal $(s,z)$-paths $P$.
\begin{itemize}
\item $P$ is a \emph{shortest} temporal $(s,z)$-path if there is no temporal path~$P'$ from $s$ to $z$ such that the length of $P'$ is strictly less than the length of $P$.
\item $P$ is a \emph{foremost} temporal $(s,z)$-path if there is no temporal path~$P'$ from $s$ to $z$ such that $P'$ arrives at $z$ at a strictly smaller time than $P$.
\item $P$ is a \emph{fastest} temporal $(s,z)$-path if there is no temporal path $P'$ from $s$ to $z$ such that the difference between the time at which $P'$ starts at $s$ and the time at which $P'$ arrives at $z$ is strictly smaller than the analogous difference of times for $P$.
\end{itemize}

\subsection{Temporal Betweenness Centrality}
We follow the notation and definition for temporal betweenness given by \citet{BMNR20}.
Let $\mathcal{G}=(V,\mathcal{E},T)$ be a temporal graph. For any $s,z \in V$, $\sigmageneric{sz}$ is the number of $\star$-optimal temporal paths from~$s$ to~$z$.
We define $\sigmageneric{vv}:=1$.  For any vertex $v \in V$, we write $\sigmageneric{sz}(v)$ for the number of $\star$-optimal paths that pass through $v$.
We set $\sigmageneric{sz}(s):=\sigmageneric{sz}$ and
$\sigmageneric{sz}(z):=\sigmageneric{sz}$.
We do not assume that there is a temporal path from any vertex
to any other vertex in the graph. 
To determine between which (ordered) pairs of vertices a temporal path exists, we use a \emph{connectivity matrix} $A$ of the temporal
graph: let $A$ be a
$|V|\times |V|$ matrix, where for every $v,w\in V$ we have that $A_{v,w}=1$ if
there is a temporal path from $v$ to $w$, and $A_{v,w}=0$ otherwise.
Note that $A_{s,z}=1$ if and only if $\sigmageneric{sz}\neq 0$. 
Formally, temporal betweenness based on $\star$-optimal temporal paths is defined as follows.
\begin{definition}[Temporal Betweenness]
The \emph{temporal betweenness} of any vertex~$v\in V$ is given by:
\[
	C^{(\star)}_B(v) := \sum_{s \neq v \neq z \text{ and } A_{s,z}=1}
	\frac{\sigmageneric{sz}(v)}{\sigmageneric{sz}}. 
\]
\end{definition}

Crucially for our work, it turns out that we can adapt any algorithm for \TemporalPaths into one that computes temporal betweenness with only polynomial overhead; we explain this reduction in \cref{sec:countvsbetweenness}.

In the reverse direction, a reduction by \citet{BMNR20} from \TemporalPaths to the problem of computing temporal betweenness centrality (based on foremost or fastest temporal paths) implies that our parameterised hardness result (\cref{thm:parameterizedhardness}) also holds for temporal betweenness computation, as the reduction increases the feedback vertex number of the underlying graph by at most three.

If we can only count temporal paths \emph{approximately}, however, the relationship between temporal path counting and temporal betweenness computation is not so straightforward: approximating the number of temporal $(s,z)$-paths that use a specific vertex can in general be much harder than approximating the total number of temporal $(s,z)$-paths.  This issue is discussed in more detail in \cref{sec:countvsbetweenness}.

\subsection{Temporal Betweenness vs.\ Temporal Path Counting}\label{sec:countvsbetweenness}
In this subsection we discuss the relationship between the problems of computing temporal betweenness and counting temporal paths. We show that we can compute temporal betweenness based on foremost and fastest temporal paths using an algorithm for \TemporalPaths with only polynomial overhead in the running time.
Let $\mathcal{G}=(V,\mathcal{E},T)$ be a temporal graph. We start with the following easy observation.
\begin{observation}\label{obs:basic1}
Given an algorithm to count all $\star$-optimal temporal $(s,z)$-paths in $\mathcal{G}$ in time $t(\mathcal{G})$, we can compute the temporal betweenness based on $\star$-optimal temporal paths of any vertex of $\mathcal{G}$ in $t(\mathcal{G})\cdot |\mathcal{G}|^{\mathcal{O}(1)}$ time.
\end{observation}
This follows by observing that we can count the number of temporal $(s,z)$-paths in $\mathcal{G}$ that visit a vertex $v$ by first counting all temporal $(s,z)$-paths in $\mathcal{G}$ and then subtracting the number of temporal $(s,z)$-paths in $\mathcal{G}-\{v\}$.

Next we observe that we can count foremost and fastest temporal $(s,z)$-paths using an algorithm for \TemporalPaths, with only polynomial overhead.
\begin{observation}\label{obs:basic2}
Given an algorithm for \TemporalPaths that runs in time $t(\mathcal{G})$, we can compute all foremost temporal $(s,z)$-paths and all fastest temporal $(s,z)$-paths in $t(\mathcal{G})\cdot |\mathcal{G}|^{\mathcal{O}(1)}$ time.
\end{observation}

First, note that we can compute a foremost temporal $(s,z)$-path and a fastest temporal $(s,z)$-path in polynomial time~\cite{xuan_computing_2003,wu_efficient_2016}. In the case of foremost temporal $(s,z)$-paths, we can in this way obtain the time at which a foremost temporal $(s,z)$-path arrives at $z$ and remove all time-edges with later time labels from $\mathcal{G}$. After this modification, every temporal $(s,z)$-path is foremost hence we can count them using an algorithm for \TemporalPaths.

In the case of fastest temporal $(s,z)$-paths, we can in the same way obtain the time difference $t_f$ between starting at $s$ and arriving at $z$ for any fastest temporal $(s,z)$-path. We can now iterate over all intervals $[t_0,t_0+t_f]$ with $1\le t_0\le T-t_f$ and, for each one, create an instance of \TemporalPaths by removing all time-edges from $\mathcal{G}$ that are either earlier than $t_0$ or later than $t_0+t_f$. After this modification, every temporal $(s,z)$-path in the instance corresponding to any interval is fastest, and every fastest temporal path survives in exactly one instance; hence we can count fastest temporal paths by calling an algorithm for \TemporalPaths on each instance and summing the results.

Using \cref{obs:basic1,obs:basic2} we obtain the following lemma, which implies that our polynomial-time and FPT-algorithms for special cases of \TemporalPaths yield polynomial-time solvability and fixed-parameter tractability results respectively for temporal betweenness based on foremost temporal paths or fastest temporal paths, under the same restrictions.
\begin{lemma}\label{cor:count-to-betweenness}
Given an algorithm for \TemporalPaths that runs in time $t(\mathcal{G})$, we can compute the temporal betweenness based on foremost temporal paths or fastest temporal paths of any vertex of $\mathcal{G}$ in $t(\mathcal{G})\cdot |\mathcal{G}|^{\mathcal{O}(1)}$ time.
\end{lemma}

If we can only count temporal paths \emph{approximately}, however, the relationship between temporal path counting and temporal betweenness computation is not so straightforward.  In the exact setting, we were able to determine the number of temporal $(s,z)$-paths visiting $v$ by calculating the difference between the number of temporal $(s,z)$-paths in $\mathcal{G}$ and $\mathcal{G} - \{v\}$ respectively.  However, in the approximate setting, we cannot use the same strategy: if there are $N$ temporal paths in total and $N_{-v}$ is an $\varepsilon$-approximation to the number of temporal paths that do not contain $v$, it does not follow that $N - N_{-v}$ is an $\varepsilon$-approximation to the number of temporal paths that do contain $v$, as the relative error will potentially be much higher if the proportion of temporal paths containing $v$ is very small.  A similar issue arises if we aim to estimate the number of temporal paths through $v$ by sampling a collection of temporal paths (from an approximately uniform distribution) and using the proportion that contain $v$ as an estimate for the total proportion of temporal paths containing $v$: if the proportion that contain $v$ is exponentially small, we would need exponentially many samples to have a non-trivial probability of finding at least one temporal path which does contain $v$; otherwise we deduce incorrectly that there are no temporal paths through $v$ and output $0$, which cannot be an $\varepsilon$-approximation of a non-zero number of temporal paths for any $\varepsilon < 1$.

Lastly, we briefly shift our attention to computational hardness. \citet{BMNR20} provide a reduction from \TemporalPaths to the computation of temporal betweenness based on foremost temporal paths and to the computation of temporal betweenness based on fastest temporal paths. In both cases, three new vertices are added to the temporal graph and all newly added time-edges are incident with at least one of the newly added vertices. This implies that our parameterised hardness result in the next section (\cref{thm:parameterizedhardness}) also holds for temporal betweenness computation based on foremost temporal paths or fastest temporal paths, since the reductions by \citet{BMNR20} increase the feedback vertex number of the underlying graph by at most three.

\subsection{Parameterised Counting Complexity}\label{sec:paramcompl}
We use standard definitions and terminology from parameterised complexity theory~\cite{DF13,FG06,Cyg+15}.
A parameterised counting problem $F,\kappa$ is in FPT (or \emph{fixed-parameter tractable}) if there is an algorithm that solves any instance~$(I,k)$ of $F,\kappa$ in $f(k)\cdot |I|^{\mathcal{O}(1)}$ time for some computable function $f$~\cite{flum2004parameterized,FG06}.
A parameterised counting Turing reduction from $F,\kappa$ to $F',\kappa'$ is an algorithm with oracle access to $F',\kappa'$ that solves any instance $(I,k)$ of $F,\kappa$ in $f(k)\cdot |I|^{\mathcal{O}(1)}$ time, where for all instances $(I',k')$ of $F',\kappa'$ queried to the oracle of $F',\kappa'$ we have that $k'\le g(k)$, for some computable functions $f,g$.
A parameterised counting problem $F,\kappa$ is hard for \#W[1] if there is a parameterised counting Turing reduction from \textsc{\#Multicoloured Clique} parameterised by the number of colours to~$F,\kappa$~\cite{flum2004parameterized,FG06}; in \textsc{\#Multicoloured Clique} we are given a $k$-partite graph and are asked to count the number of $k$-cliques.
A parameterised counting problem $F,\kappa$ is hard for $\oplus$W[1] (``parity-W[1]'') it there is a parameterised counting Turing reduction from \textsc{$\oplus$Multicoloured Clique} parameterised by the number of colours to~$F,\kappa$~\cite{bjorklund2015parity};
 in \textsc{$\oplus$Multicoloured Clique} we are given a $k$-partite graph and are asked to count the number of $k$-cliques modulo two, that is, decide whether the number of $k$-cliques is odd.

If a \#W[1]-hard (resp.\ $\oplus$W[1]-hard) parameterised counting problem $F,\kappa$ admits an FPT-algorithm, then \#W[1]$=$FPT (resp.\ $\oplus$W[1]$=$FPT), which is generally not believed to be the case~\cite{flum2004parameterized,bjorklund2015parity,FG06}. 
We remark that \#W[1]$=$FPT clearly implies W[1]$=$FPT (and also $\oplus$W[1]$=$FPT), whereas at the time of writing it is unknown whether $\oplus$W[1]$=$FPT implies W[1]$=$FPT.

\subsection{Approximate Counting and Sampling}\label{sec:approxcompl}

Many computational problems can be associated with a relation $R \subseteq \Sigma^* \times \Sigma^*$, where $\Sigma$ is some finite alphabet and $R$ can be seen as assigning to each problem instance $x \in \Sigma^*$ a set of solutions (namely the set $\sol_R(x) := \{y \in \Sigma^* \colon xRy\}$).  For a given relation $R \subseteq \Sigma^* \times \Sigma^*$ and an instance $x \in \Sigma^*$, we might be interested in the associated decision problem (``is $\sol_R(x)$ non-empty?''), counting problem (``determine $|\sol_R(x)|$'') or the uniform generation problem (``return a uniformly random element of $\sol_R(x)$'').  

We begin by defining our notion of efficient approximation for counting problems (see, for example, \cite[Chapter~11]{ProbComp}).

\begin{definition}
Let $F: \Sigma^* \rightarrow \mathbb{N} \cup \{0\}$ be a counting problem.  A fully polynomial randomised approximation scheme (FPRAS) for $F$ is a randomised approximation scheme that takes an instance $I$ of $F$ (with $|I| = n$), and real numbers $\varepsilon > 0$ and $0 < \delta < 1$, and in time $\poly(n,1/\varepsilon,\log(1/\delta))$ outputs a rational number $z$ such that
\[
\mathbb{P}[(1-\varepsilon)F(I) \le z \le (1 + \varepsilon)F(I)] \geq 1 - \delta.
\]
\end{definition}

We are also interested in the parameterised analogue of an FPRAS, a \emph{fixed parameter tractable randomised approximation scheme (FPTRAS)} \cite{arvindraman02}. 

\begin{definition}
Let $F: \Sigma^* \rightarrow \mathbb{N} \cup \{0\}$ be a counting problem with parameterisation $\kappa: \Sigma^* \rightarrow \mathbb{N}$.  An FPTRAS for $(F,\kappa)$ is a randomised approximation scheme that takes an instance $I$ of $F$ (where $|I| = n$), and real numbers $\varepsilon > 0$ and $0 < \delta < 1$, and in time $f(\kappa(I)) \cdot \poly(n,1/\varepsilon,\log(1/\delta))$ (where $f$ is any computable function) outputs a rational number $z$ such that
\[
\mathbb{P}[(1-\varepsilon)F(I) \leq z \leq (1 + \varepsilon)F(I)] \geq 1 - \delta.
\]
\end{definition}

For convenience, we often refer to a number $z$ satisfying $(1 - \varepsilon)F(I) \le z \le (1 + \varepsilon)F(I)$ as an \emph{$\varepsilon$-approximation} to $|F(I)|$.

It is well-known that there is a close relationship between the algorithmic problems of approximately counting solutions and generating solutions almost uniformly \cite{jerrum86generation}.  We will make use of this relationship when considering approximating the temporal betweenness of a vertex in Section \ref{sec:approx-between}, and to do so need a notion of efficient almost uniform sampling \cite[Chapter~11]{ProbComp}.

\begin{definition}
Let $S$ be the uniform generation problem associated with the relation $R \subseteq \Sigma^* \times \Sigma^*$.  A \emph{fully polynomial almost uniform sampler (FPAUS)} for $S$ is a randomised algorithm which takes as input an instance $x$ of $S$ (with $|x| = n$) together with an error parameter $0 \le \delta \le 1$, and in time $\poly(n, \log(1/\delta))$ returns an element of $\sol_R(x) = \{y \in \Sigma^* \colon xRy\}$; for any fixed $y$ with $xRy$, the probability $p_y$ that the algorithm returns $y$ satisfies $(1 - \delta)/|\sol_R(x)| \le p_y \le (1 + \delta)/|\sol_R(x)|$.  
\end{definition}

One can naturally define the parameterised analogue of an FPAUS.

\begin{definition}
Let $S$ be the uniform generation problem associated with the relation $R \subseteq \Sigma^* \times \Sigma^*$, and let $\kappa: \Sigma^* \rightarrow \mathbb{N}$ be a parameterisation of $S$.  A \emph{fixed parameter tractable almost uniform sampler (FPTAUS)} for $S$ is a randomised algorithm which takes as input an instance $x$ of $S$ (with $|x| = n$) together with an error parameter $0 \le \delta \le 1$, and in time $f(\kappa(x)) \cdot \poly(n, \log(1/\delta))$ (where $f$ is any computable function) returns an element of $\sol_R(x) = \{y \in \Sigma^* \colon xRy\}$; for any fixed $y$ with $xRy$, the probability $p_y$ that the algorithm returns $y$ satisfies $(1 - \delta)/|\sol_R(x)| \le p_y \le (1 + \delta)/|\sol_R(x)|$.  
\end{definition}

\section{Intractability Results for Temporal Path Counting}\label{sec:intractability}

In this section we prove two hardness results for \TemporalPaths.  In \cref{sec:param-hard} we demonstrate parameterised intractability with respect to the feedback vertex number of the underlying graph.  We follow this in \cref{sec:approx-hard} with an easy reduction demonstrating that the classical \#P-complete \Paths problem~\cite{valiant_complexity_1979} (definition as below) is unlikely to admit an FPRAS in general, which straightforwardly implies the same result for \TemporalPaths.
\problemdef{\Paths}{A graph $G = (V,E)$ and two vertices $s,z \in V$.}{Compute the number of paths from $s$ to $z$ in $G$.}

\subsection{Parameterised Hardness}\label{sec:param-hard}

In this section we present our main parameterised hardness result, which provides strong evidence that \TemporalPaths does not admit an FPT algorithm when parameterised by the feedback vertex number of the underlying graph. Note that this also rules out FPT algorithms for many other parameterizations, including the treewidth of the underlying graph. However, it is folklore that \Paths admits an FPT algorithm parameterised by treewidth as a parameter (this is also implied by our result \cref{thm:treewidth}). The result here, therefore, means that \TemporalPaths is strictly harder than \Paths in terms of parameterised complexity for the parameterisations that are at most the feedback vertex number of the underlying graph and at least the treewidth of the underlying graph.

\begin{theorem}
\label{thm:parameterizedhardness}
\TemporalPaths is $\oplus$W[1]-hard when parameterised by the feedback vertex number of the underlying graph.
\end{theorem}

\begin{proof}
We present a parameterised counting Turing reduction from \textsc{$\oplus$Multicoloured Independent Set on 2-Track Interval Graphs} parameterised by the number of colours $k$. 
In \textsc{$\oplus$Multicoloured Independent Set on 2-Track Interval Graphs} we are given a set $I$ of interval pairs and a colouring function $c:I\rightarrow [k]$ and asked whether there is an odd number of $k$-sized sets of interval pairs in $I$ such that in each set, every two interval pairs have different colours and are non-intersecting.
Two interval pairs $([x_a,x_b],[x_{a'},x_{b'}]),([y_a,y_b],[y_{a'},y_{b'}])$ are considered non-intersecting if $[x_a,x_b]\cap[y_a,y_b]=\emptyset$ and $[x_{a'},x_{b'}]\cap[y_{a'},y_{b'}]=\emptyset$.

Inspecting the W[1]-hardness proof by~\citet{jiang2010parameterized} for \textsc{Independent Set on 2-Track Interval Graphs} shows that the reduction used from \textsc{Multicoloured Clique} parameterised by the number of colours $k$ is parsimonious\footnote{Informally speaking, parsimonious reductions do not change the number of solutions.} and
the reduction also shows W[1]-hardness for the multicoloured version of the problem. Since \textsc{$\oplus$Multicoloured Clique} is $\oplus$W[1]-hard when parameterised by the number of colours $k$~\cite{bjorklund2015parity}, we can conclude that \textsc{$\oplus$Multicoloured Independent Set on 2-Track Interval Graphs} is $\oplus$W[1]-hard when parameterised by the number of colours $k$.

Given an instance $(I,c)$ of \textsc{$\oplus$Multicoloured Independent Set on 2-Track Interval Graphs}, where $I$ is a set of interval pairs and $c:I\rightarrow [k]$ is a colouring function, we create $\mathcal{O}(2^k)$ temporal graphs. We assume w.l.o.g.\ that for all $([x_a,x_b],[x_{a'},x_{b'}]),([y_a,y_b],[y_{a'},y_{b'}])\in I$ that
$|\{x_a,x_b,y_a,y_b\}|=4$ and $|\{x_{a'},x_{b'},y_{a'},y_{b'}\}|=4$ or $([x_a,x_b],[x_{a'},x_{b'}])=([y_a,y_b],[y_{a'},y_{b'}])$, that is, if two interval pairs are different, we assume that all endpoints on each track are pairwise different.
Furthermore, we assume w.l.o.g.\ that all intervals contained in pairs in $I$ are integer subsets of $[2|I|]$. The main intuition of our construction follows:
\begin{itemize}
    \item We model track one with a path in the underlying graph and track two with time.
    \item Through the feedback vertices of the underlying graph, a temporal path can ``enter'' and ``leave'' the path that models track one. 
    \item The number of feedback vertices corresponds to the number of colours. 
    \item We have to make sure that we can determine the parity of the number of temporal paths visiting all feedback vertices.
    \item The number of temporal paths that do not correspond to independent sets should not be considered. It seems difficult to get an exact handle on the number of such paths, however we will show that this number
    is even. Note that, intuitively, this is the main reason we show hardness for $\oplus$W[1] and not \#W[1].
\end{itemize}

We construct a family of \emph{directed} temporal graphs $(\mathcal{G}^{\mathcal{C}}=(V,\mathcal{A}^{\mathcal{C}},2|I|+1))_{\mathcal{C}\subseteq [k]}$ with rational time labels (such that the maximum time label is at most $2|I|+1$), where $\mathcal{A}^{\mathcal{C}}\subseteq V\times V\times \mathbb{Q}$ for all $\mathcal{C}\subseteq [k]$. Towards the end of the proof we explain how to remove the need for directed edges which will also have the consequence that the temporal graphs only contain strict temporal paths. Note that we can scale up the lifetime to remove the need for rational time labels, however using rational time labels will be convenient in the construction and the correctness proof.

\begin{itemize}
    \item We set $V:=V_I\cup\{s,z',z\}\cup \{w_{1},\ldots,w_{k}\}\cup\{u_x\mid x\in I\}$, where $V_I:=\{v_1, \ldots, v_{2|I|}\}$.
    \item We set $\mathcal{A}^{\mathcal{C}}:=\bigcup_{x\in I\wedge c(x)\in\mathcal{C}} \mathcal{A}_x\cup \{(s,w_{i},1)\mid i\in[k]\}\cup\{(z',z,2|I|+1)\}$, where 
    \begin{align*}
        \mathcal{A}_x:= & \{(w_{c(x)}, u_x,a'), (u_x,v_a,b'), (u_x,v_a,b'+1-a\varepsilon),(v_b,z',b')\}\\
        &\cup \{(v_b,w_{i},b')\mid i\in[k]\wedge i\neq c(x)\}\\
        &\cup\{(v_j,v_{j+1},b'),(v_j,v_{j+1},b'+1-(j+1)\varepsilon)\mid j\in\{a, \ldots, b-1\}\}
        \end{align*} 
    for $x=([a,b],[a',b'])\in I$ and $\varepsilon = \frac{1}{2|I|}$.
\end{itemize}
For all $\mathcal{G}^{\mathcal{C}}$ we use $s$ as the starting vertex and $z$ as the end vertex of the temporal paths we want to count. 
The temporal graphs $\mathcal{G}^\mathcal{C}$ can each clearly be constructed in polynomial time and it is easy to see that the vertex set $\{s,z',z,w_1,\ldots,w_k\}$ constitutes a feedback vertex set of size $\mathcal{O}(k)$ for each of them (even if edge directions are removed). The construction is illustrated in \cref{fig:hardness1}. 

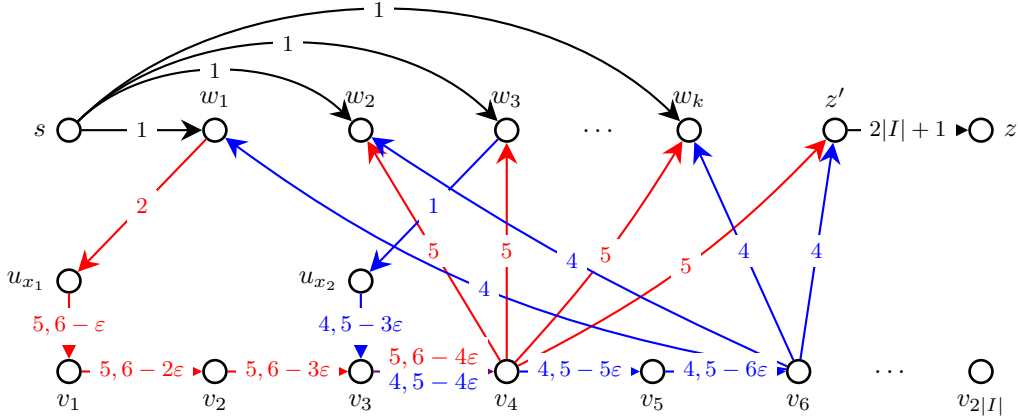
\begin{figure}[t]
\begin{center}
\begin{tikzpicture}[line width=1pt, scale=.8, xscale=1.2]

    \node[vert,label=below:$v_1$] (V1) at (0,0) {}; 
    \node[vert,label=below:$v_2$] (V2) at (2,0) {}; 
    \node[vert,label=below:$v_3$] (V3) at (4,0) {}; 
    \node[vert,label=below:$v_4$] (V4) at (6,0) {}; 
    \node[vert,label=below:$v_5$] (V5) at (8,0) {}; 
    \node[vert,label=below:$v_6$] (V6) at (10,0) {};
    \node (A) at (11.25,0) {$\ldots$};
    \node[vert,label=below:$v_{2|I|}$] (Vi) at (12.5,0) {};
    
    \node[vert,label=left:$u_{x_1}$] (U1) at (0,1.5) {}; 
    \node[vert,label=left:$u_{x_2}$] (U2) at (4,1.5) {}; 
    
    \node[vert,label=left:$s$] (S) at (0,4) {}; 
    \node[vert,label=above:$z'$] (Z1) at (10.5,4) {};
    \node[vert,label=right:$z$] (Z2) at (12.5,4) {};
    \node[vert,label=above:$w_1$] (W1) at (2,4) {};
    \node[vert,label=above:$w_2$] (W2) at (4,4) {};
    \node[vert,label=above:$w_3$] (W3) at (6,4) {};
    \node (B) at (7.25,4) {$\ldots$};
    \node[vert,label=above:$w_k$] (Wk) at (8.5,4) {};

	\draw[diredge] (S) --node[timelabel] {$1$} (W1);
	\draw[diredge] (S) edge[bend left=50] node[timelabel] {$1$} (W2);
	\draw[diredge] (S) edge[bend left=50] node[timelabel] {$1$} (W3);
	\draw[diredge] (S) edge[bend left=50] node[timelabel] {$1$} (Wk);
	\draw[diredge] (Z1) --node[timelabel] {$2|I|+1$} (Z2);

	\draw[diredge,red] (W1) --node[timelabel] {$2$} (U1);
	\draw[diredge,red] (U1) --node[timelabel] {$5,6-\varepsilon$} (V1);
	\draw[diredge,red] (V4) edge[bend right=10] node[timelabel] {$5$} (Z1);
	\draw[diredge,red] (V4) --node[timelabel] {$5$} (W2);
	\draw[diredge,red] (V4) --node[timelabel] {$5$} (W3);
	\draw[diredge,red] (V4) edge[bend right=5] node[timelabel] {$5$} (Wk);
	\draw[diredge,red] (V1) --node[timelabel] {$5,6-2\varepsilon$} (V2);
	\draw[diredge,red] (V2) --node[timelabel] {$5,6-3\varepsilon$} (V3);
	\draw[diredge,violet] (V3) --node[timelabel] {\begin{tabular}{@{}c@{}}\textcolor{red}{$5,6-4\varepsilon$}\\ \textcolor{blue}{$4,5-4\varepsilon$}\end{tabular}} (V4);
	
	\draw[diredge,blue] (W3) --node[timelabel] {$1$} (U2);
	\draw[diredge,blue] (U2) --node[timelabel] {$4,5-3\varepsilon$} (V3);
	\draw[diredge,blue] (V6) --node[timelabel] {$4$} (Z1);
	\draw[diredge,blue] (V6) edge[bend left=15] node[timelabel] {$4$} (W1);
	\draw[diredge,blue] (V6) edge[bend left=5] node[timelabel] {$4$} (W2);
	\draw[diredge,blue] (V6) --node[timelabel] {$4$} (Wk);
	\draw[diredge,blue] (V4) --node[timelabel] {$4,5-5\varepsilon$} (V5);
	\draw[diredge,blue] (V5) --node[timelabel] {$4,5-6\varepsilon$} (V6);
\end{tikzpicture}
    \end{center}
    \caption{Illustration of $\mathcal{G}^\mathcal{C}$ with $\mathcal{C}=\{1,3\}$ and two interval pairs $x_1,x_2\in I$ where $x_1=([1,4],[2,5])$ and $x_2=([3,6],[1,4])$, and the corresponding colours are $c(x_1)=1$ and $c(x_2)=3$. The arcs added for $x_1$ are depicted in red and the arcs added for $x_2$ are depicted in blue.}\label{fig:hardness1}
\end{figure}

We now show some properties of the construction that will help to prove correctness of the reduction.
\begin{claim}\label{claim:unique}
Let $x=([a,b],[a',b'])\in I$, $i\in[k]\wedge i\neq c(x)$, and $\mathcal{C}\subseteq[k]$ with $c(x)\in \mathcal{C}$. Then there is exactly one temporal $(w_{c(x)},w_i)$-path $P$ in $\mathcal{G}^\mathcal{C}$ such that $V(P)\cap V_I=\{v_a,v_{a+1},\ldots,v_b\}$. Furthermore, there is exactly one temporal $(w_{c(x)},z')$-path $P$ in $\mathcal{G}^\mathcal{C}$ such that $V(P)\cap V_I=\{v_a,v_{a+1},\ldots,v_b\}$.
\end{claim}
\begin{claimproof}[Proof of claim]
Let $x=([a,b],[a',b'])\in I$, $i\in[k]\wedge i\neq c(x)$, and $\mathcal{C}\subseteq[k]$ with $c(x)\in \mathcal{C}$. We first show that there is a temporal $(w_{c(x)},w_i)$-path $P$ in $\mathcal{G}^\mathcal{C}$ such that $V(P)\cap V_I=\{v_a,v_{a+1},\ldots,v_b\}$. Then we show that the path is unique. The case for a temporal $(w_{c(x)},z')$-path works analogously.

Consider the path $P=((w_{c(x)},u_x,a'),(u_x,v_a,b'),(v_a,v_{a+1},b'),\ldots, (v_{b-1},v_b,b'),(v_b,w_{i},b'))$. It is easy to verify that this path is contained in $\mathcal{G}^\mathcal{C}$ and that we have $V(P)\cap V_I=\{v_a,v_{a+1},\ldots,v_b\}$. Now assume for contradiction that there is a temporal $(w_{c(x)},w_i)$-path $P'$ in $\mathcal{G}^\mathcal{C}$ such that $V(P')\cap V_I=\{v_a,v_{a+1},\ldots,v_b\}$ and $P'\neq P$. Recall that we assume that if two interval pairs $x,y\in I$ are different, then all endpoints on each track are pairwise different. This means that, by construction of $\mathcal{G}^\mathcal{C}$, the temporal path $P'$ has to visit the same vertices as $P$ in the same order. This implies that $P'$ contains a transition $(u,v,t)$ such that $P$ contains the transition $(u,v,t')$ with some $t'\neq t$. However, note that the only time step when we have a transition from $v_b$ to $w_i$ is $b'$ and also the earliest time step when we can arrive at $v_a$ (via $u_x$) is $b'$. Hence we must have that $t=b'$ unless $(u,v,t)=(w_{c(x)},u_x,a')$. In both cases we obtain a contradiction to the assumption that $P'\neq P$.
\end{claimproof}

\begin{claim}\label{claim:cheating}
Let $V^\star=\{v_a,v_{a+1},\ldots,v_b\}$ for some $[a,b]$ such that for all $a',b'$ we have that $([a,b],[a',b'])\notin I$. Then for all $i,j\in [k]$ and for all $\mathcal{C}\subseteq[k]$ we have that $\mathcal{G}^\mathcal{C}$ contains an even number of temporal $(w_i,w_j)$-paths $P$ such that $V(P)\cap V_I=V^\star$.
Furthermore, there is an even number of temporal $(w_i,z')$-paths $P$ in $\mathcal{G}^\mathcal{C}$ such that $V(P)\cap V_I=V^\star$.
\end{claim}
\begin{claimproof}[Proof of claim]
Let $V^\star=\{v_a,v_{a+1},\ldots,v_b\}$ for some $[a,b]$ such that for all $a',b'$ we have that $([a,b],[a',b'])\notin I$. We show that for all $i,j\in [k]$ and for all $\mathcal{C}\subseteq[k]$ we have that $\mathcal{G}^\mathcal{C}$ contains an even number of temporal $(w_i,w_j)$-paths $P$ such that $V(P)\cap V_I=V^\star$. The case for temporal $(w_i,z')$-paths works analogously. 

If there is no temporal path satisfying the conditions, then we are done. Hence, assume that there is a $\mathcal{C}\subseteq[k]$ and $i,j\in [k]$ such that $\mathcal{G}^\mathcal{C}$ contains a temporal $(w_i,w_j)$-path $P$ with $V(P)\cap V_I=V^\star$. Let $\mathcal{P}_{i,j}^\mathcal{C}$ denote the set of all temporal $(w_i,w_j)$-paths $P$ in $\mathcal{G}^\mathcal{C}$ with $V(P)\cap V_I=V^\star$. We will argue that $|\mathcal{P}_{i,j}^\mathcal{C}|\equiv 0\mod 2$.

Since $\mathcal{P}_{i,j}^\mathcal{C}\neq\emptyset$ we have an interval pair $x=([a,\hat{b}],[\hat{a}',\hat{b}'])\in I$ with $c(x)=i$, otherwise the first vertex from $V_I$ that is visited by each temporal path $P\in\mathcal{P}_{i,j}^\mathcal{C}$ cannot be $v_a$. Now consider the vertex set $\hat{V}=\{v_a, v_{a+1},\ldots,v_{\min(\hat{b},b)}\}\subseteq V^\star$. By construction of $\mathcal{G}^\mathcal{C}$, we have for any two vertices $v_\ell, v_{\ell+1}\in \hat{V}$ that $(v_\ell,v_{\ell+1},\hat{b}')\in \mathcal{A}^{\mathcal{C}}$ and $(v_\ell,v_{\ell+1},\hat{b}'+1-(\ell+1)\varepsilon)\in \mathcal{A}^{\mathcal{C}}$. Furthermore, we have that the first time-arc of every temporal path in $\mathcal{P}_{i,j}^\mathcal{C}$ is $(w_i,u_x,\hat{a}')$ and we have $(u_x,v_a,\hat{b}')\in \mathcal{A}^{\mathcal{C}}$ and $(u_x,v_a,\hat{b}'+1-a\varepsilon)\in \mathcal{A}^{\mathcal{C}}$. We call the collection of these time-arcs the \emph{early lane} $\hat{\mathcal{A}}_{i,j}^\mathcal{C}$ of $\mathcal{P}_{i,j}^\mathcal{C}$, formally $\hat{\mathcal{A}}_{i,j}^\mathcal{C}=\{(u_x,v_a,\hat{b}'),(u_x,v_a,\hat{b}'+1-a\varepsilon)\}\cup \{(v_\ell,v_{\ell+1},\hat{b}'),(v_\ell,v_{\ell+1},\hat{b}'+1-(\ell+1)\varepsilon)\mid v_\ell, v_{\ell+1}\in \hat{V}\}$. Notice that every static arc that appears as a time-arc in $\hat{\mathcal{A}}_{i,j}^\mathcal{C}$ appears exactly twice (i.e.\ with two different time labels). 

We now distinguish temporal $(w_i,w_j)$-paths $P$ with $V(P)\cap V_I=V^\star$ by the last vertex they visit using the early lane. Formally, given a temporal $(w_i,w_j)$-path $P$ with $V(P)\cap V_I=V^\star$, let $\hat{v}_P\in \hat{V}$ be the last vertex of the longest prefix $\hat{P}$ of $P$ such that $\hat{P}$ only consists of time-arcs in the early lane  $\hat{\mathcal{A}}_{i,j}^\mathcal{C}$. For $v_\ell\in V_I$, let $\mathcal{P}_{i,j}^\mathcal{C}(v_\ell)$ be the sets of all temporal $(w_i,w_j)$-paths $P$ with $V(P)\cap V_I=V^\star$ and $\hat{v}_P=v_\ell$; intuitively these are the temporal paths that `leave' the early lane at $v_\ell$. Note that we have $\mathcal{P}_{i,j}^\mathcal{C}=\bigcup_{v_\ell\in V_I}\mathcal{P}_{i,j}^\mathcal{C}(v_\ell)$, and since the sets $\mathcal{P}_{i,j}^\mathcal{C}(v_\ell)$ with $v_\ell\in V_I$ are by definition pairwise disjoint, we also have $|\mathcal{P}_{i,j}^\mathcal{C}|=\sum_{v_\ell\in V_I}|\mathcal{P}_{i,j}^\mathcal{C}(v_\ell)|$.

Finally, we show that for all $v_\ell\in V_I$ we have $|\mathcal{P}_{i,j}^\mathcal{C}(v_\ell)|\equiv 0\mod 2$, from which the claim then follows. Let $v_\ell\in V_I$. If $\mathcal{P}_{i,j}^\mathcal{C}(v_\ell)=\emptyset$ we are done, hence assume that $\mathcal{P}_{i,j}^\mathcal{C}(v_\ell)\neq\emptyset$. 
Recall that all temporal paths in $\mathcal{P}_{i,j}^\mathcal{C}(v_\ell)$ use time-arcs from the early lane $\hat{\mathcal{A}}_{i,j}^\mathcal{C}$ of $\mathcal{P}_{i,j}^\mathcal{C}$ until they reach $v_\ell$. We now partition $\mathcal{P}_{i,j}^\mathcal{C}(v_\ell)$ into $\mathcal{Q}_{i,j}^\mathcal{C}(v_\ell)$ and $\mathcal{R}_{i,j}^\mathcal{C}(v_\ell)$. Assume that $\ell>a$; the case that $\ell=a$ works analogously. The set $\mathcal{Q}_{i,j}^\mathcal{C}(v_\ell)$ contains all temporal paths from $\mathcal{P}_{i,j}^\mathcal{C}(v_\ell)$ that use the time-arc $(v_{\ell-1},v_{\ell},\hat{b}')$ and the set $\mathcal{R}_{i,j}^\mathcal{C}(v_\ell)$ contains all temporal paths from $\mathcal{P}_{i,j}^\mathcal{C}(v_\ell)$ that use the time-arc $(v_{\ell-1},v_{\ell},\hat{b}'+1-\ell\varepsilon)$. In both cases, the mentioned time-arc is the last time-arc from the early lane used by the temporal paths. 

We show that $|\mathcal{Q}_{i,j}^\mathcal{C}(v_\ell)|=|\mathcal{R}_{i,j}^\mathcal{C}(v_\ell)|$ by giving a \emph{bijection} $f:\mathcal{Q}_{i,j}^\mathcal{C}(v_\ell)\rightarrow \mathcal{R}_{i,j}^\mathcal{C}(v_\ell)$ between the two sets. This then implies that $|\mathcal{P}_{i,j}^\mathcal{C}(v_\ell)|\equiv 0\mod 2$. Given a temporal path $P\in \mathcal{Q}_{i,j}^\mathcal{C}(v_\ell)$, the function $f$ maps $P$ to $f(P)=P'$, where $P'$ is obtained from $P$ by replacing time-arc $(v_{\ell-1},v_{\ell},\hat{b}')$ with $(v_{\ell-1},v_{\ell},\hat{b}'+1-\ell\varepsilon)$. We first show that $P'$ is indeed a temporal path which implies that it is contained in the set $\mathcal{R}_{i,j}^\mathcal{C}(v_\ell)$. Since the new time-arc $(v_{\ell-1},v_{\ell},\hat{b}'+1-\ell\varepsilon)$ has a larger time label than the original one, the following time-arc in $P'$ cannot have a smaller time label in order for $P'$ to be a temporal path. However, note that the time-arc in $P'$ that directly follows $(v_{\ell-1},v_{\ell},\hat{b}'+1-\ell\varepsilon)$ is \emph{not} a time-arc from the early lane, hence its time label is at least $b'+1$. It follows that $P'$ is a temporal path. We now have shown that $f$ maps each element in $\mathcal{Q}_{i,j}^\mathcal{C}(v_\ell)$ to exactly one element in $\mathcal{R}_{i,j}^\mathcal{C}(v_\ell)$, meaning that $f$ is injective. To finish the proof, we consider the inverse function $f^{-1}:\mathcal{R}_{i,j}^\mathcal{C}(v_\ell)\rightarrow \mathcal{Q}_{i,j}^\mathcal{C}(v_\ell)$ of $f$, which maps temporal paths $Q\in \mathcal{R}_{i,j}^\mathcal{C}(v_\ell)$ to $f^{-1}(Q)=Q'$, where $Q'$ is obtained from $Q$ by replacing time-arc $(v_{\ell-1},v_{\ell},\hat{b}'+1-\ell\varepsilon)$ with $(v_{\ell-1},v_{\ell},\hat{b}')$. We show that $f^{-1}$ is also injective by showing that $Q'$ is a temporal path and hence element of $\mathcal{Q}_{i,j}^\mathcal{C}(v_\ell)$. Now the new time-arc in $Q'$ has a smaller time label than the original one, so the preceding time-arc cannot have a larger time label. Let $e$ be the time-arc in $Q'$ that directly precedes $(v_{\ell-1},v_{\ell},\hat{b}')$. We know that $e$ is also a time-arc from the early lane, hence its label is either also $\hat{b}'$ or $\hat{b}'+1-(\ell-1)\varepsilon$. However, if the label of $e$ was $\hat{b}'+1-(\ell-1)\varepsilon$, then it would be larger than the label of the replaced original time-arc, which had label $\hat{b}'+1-\ell\varepsilon$, a contradiction to the assumption that $Q$ is a temporal path. It follows that the label of $e$ is $\hat{b}'$ which implies that $Q'$ is a temporal path. This shows that also $f^{-1}$ is injective which means that $f$ is indeed a bijection.
\end{claimproof}

We call an independent set \emph{colourful} if every vertex of the independent set has a different colour.
\begin{claim}\label{claim:count}
Let $\mathcal{C}\subseteq [k]$. The number of temporal $(s,z)$-paths in $\mathcal{G}^\mathcal{C}$ is even if and only if the number of colourful
independent sets $X$ in $(I,c)$ with 
$\{c(x) \mid x \in X\}
\subseteq\mathcal{C}$ is even.
\end{claim}
\begin{claimproof}[Proof of claim]
Let $\mathcal{C}\subseteq [k]$ and let $\mathcal{P}^\mathcal{C}$ be the set of all temporal $(s,z)$-paths in $\mathcal{G}^\mathcal{C}$. 
We partition the set $\mathcal{P}^\mathcal{C}$ into two parts, $\mathcal{Q}^\mathcal{C}$ and $\mathcal{R}^\mathcal{C}$. Intuitively, the set $\mathcal{Q}^\mathcal{C}$ will contain ``cheating'' temporal $(s,z)$-paths, that do not correspond to a colourful independent set in instance $(I,c)$ that uses only colours in $\mathcal{C}$. We show that $|\mathcal{Q}^\mathcal{C}|\equiv 0\mod 2$. The set $\mathcal{R}^\mathcal{C}$ will contain temporal $(s,z)$-paths that have a one-to-one correspondence with colourful independent sets in $(I,c)$ that use colours in $\mathcal{C}$.
We show that $|\mathcal{R}^\mathcal{C}|=|\{X\mid X \text{ is a colourful independent set in } (I,c) \text{ and each element in } X \text{ has a colour from }\mathcal{C}\}|$. 
From this the claim then follows.

Let $P\in \mathcal{P}^\mathcal{C}$. If for some $i,j\in\mathcal{C}$ the temporal path $P$ contains a temporal subpath $P'$ from $w_i$ to $w_j$ or from $w_i$ to $z'$ such that $V(P')\cap V_I=\{v_a,v_{a+1},\ldots,v_b\}$ for some $[a,b]$ and for all $a',b'$ we have that $([a,b],[a',b'])\notin I$, then we define $P$ to be contained in $\mathcal{Q}^\mathcal{C}$. Otherwise, $P$ is contained in $\mathcal{R}^\mathcal{C}$. By \cref{claim:cheating} we have that $|\mathcal{Q}^\mathcal{C}|\equiv 0\mod 2$. It follows that $|\mathcal{P}^\mathcal{C}|\equiv |\mathcal{R}^\mathcal{C}|\mod 2$.

Let $P\in \mathcal{R}^\mathcal{C}$. Then we can construct a colourful independent set in $(I,c)$ that uses colours in $\mathcal{C}$ as follows. We identify the following temporal subpaths of $P$: 
\begin{itemize}
    \item $P_{i,j}$ is the subpath from $w_i$ to $w_j$, where we assume that no other vertex in $\{w_1,\ldots,w_k\}$ is visited by $P_{i,j}$, that is, $V(P_{i,j})\cap\{w_1,\ldots,w_k\}=\{w_{i},w_{j}\}$.
    \item $P_{i,z'}$ is the subpath from the ``last'' $w_i$ visited by $P$ to $z'$, that is, $V(P_{i,z'})\cap\{w_1,\ldots,w_k\}=\{w_i\}$.
\end{itemize}
For each of these subpaths $P_{i,j}$ we know that there exists an $x=([a,b],[a',b'])\in I$ such that $V(P_{i,j})\cap V_I=\{v_a,v_{a+1},\ldots,v_b\}$, otherwise we would have $P\in \mathcal{Q}^\mathcal{C}$. By construction of $\mathcal{G}^\mathcal{C}$ we also know that $c(x)=i$. Analogously we can make the same observation for  $P_{i,z'}$. Hence, we can identify each of the above defined subpaths with an element $x$ of $I$. Let $X$ denote the set of those elements. By construction we know that $X$ is colourful and uses colours from $\mathcal{C}$; what remains to show is that it is an independent set in $(I,c)$. Assume for contradiction that $X$ is not an independent set in $(I,c)$. Then there are two distinct interval pairs $x=([a,b],[a',b'])\in X$ and $x'=([c,d],[c',d'])\in X$ such that $[a,b]\cap[c,d]\neq\emptyset$ or $[a',b']\cap[c',d']\neq\emptyset$. Let $P_{i,j}$ and $P_{i',j'}$ be the two subpaths corresponding to $x$ and $x'$, respectively (here we allow $j=z'$ and $j'=z'$).
\begin{itemize}
    \item If $[a,b]\cap[c,d]\neq\emptyset$, then $V(P_{i,j})\cap V(P_{i',j'})\neq\emptyset$. This is a contradiction to the assumption that $P$ is a temporal \emph{path}, since vertices are visited multiple times.
    \item If $[a',b']\cap[c',d']\neq\emptyset$, then $P$ contains the time-arcs $(w_{c(x)}, u_x,a'), (u_x,v_a,b')$ consecutively and also time-arcs $(w_{c(x')}, u_{x'},c'), (u_{x'},v_c,d')$ consecutively (since $P$ contains $P_{i,j}$ and $P_{i',j'}$). This is a contradiction to the assumption that $P$ is a \emph{temporal} path, since the time labels on the time-arcs are not non-decreasing. 
\end{itemize}
Hence, we can conclude that $X$ is a colourful independent set in $(I,c)$ that uses colours from $\mathcal{C}$. Furthermore, by \cref{claim:unique} we have that for different temporal paths $P,P'\in \mathcal{R}^\mathcal{C}$, we obtain different colourful independent sets $X,X'$ in $(I,c)$ using colours from $\mathcal{C}$, that is, we have described an injective mapping from the temporal paths in $\mathcal{R}^\mathcal{C}$ to the colourful independent sets in $(I,c)$ that use colours from $\mathcal{C}$.

Now let $X$ be a colourful independent set in $(I,c)$ that uses colours from $\mathcal{C}$. Then we can construct a temporal $(s,z)$-path $P$ in $\mathcal{G}^\mathcal{C}$ with $P\in \mathcal{R}^\mathcal{C}$ as follows. Let $X=\{x_1, x_2,\ldots,x_{|X|}\}$ such that for all $x_i=([a,b],[a',b'])\in X$ and $x_j=([c,d],[c',d'])\in X$ we have $i<j$ if and only if $b'<c'$. Note that this indexing is well-defined since $X$ is an independent set in $(I,c)$. For each $x_i=([a,b],[a',b'])\in X$ with $i<|X|$ we define a path segment $P_i$ as follows. The path segment $P_i$ starts at $w_{c(x_i)}$, ends at $w_{c(x_{i+1})}$, and visits vertices $u_{x_i}$ and $v_a, v_{a+1},\ldots, v_b$ in $V_I$. It does so using the following time-arcs: $(w_{c(x_i)}, u_{x_i},a'), (u_{x_i},v_a,b'), (v_a,v_{a+1},b'), \ldots, (v_{b-1},v_b,b'), (v_b,w_{c(x_{i+1})},b')$. Note that all mentioned time-arcs are present in $\mathcal{G}^\mathcal{C}$ by construction and since $a'\le b'$ we have that $P_i$ is a temporal path segment. The path segment $P_{|X|}$ starts at $w_{c(x_{|X|})}$, ends at $z'$, and is defined analogously. Now the temporal path $P$ starts with the time-arc $(s,w_{x_1},1)$, then uses $P_1,P_2,\ldots,P_{|X|}$, and finally ends with time-arc $(z',z,2|I|+1)$. Assume for contradiction that $P$ is not a temporal path.
\begin{itemize}
    \item Note that by construction, $P$ visits each vertex from $\{u_{x_1},u_{x_2},\ldots,u_{x_|X|}\}$ at most once. If $P$ visits a vertex from $\{w_1, \ldots, w_k\}$ multiple times, then we have a contradiction to $X$ being colourful. If $P$ visits a vertex from $V_I$ multiple times, then we have a contradiction to $X$ being an independent set in $(I,c)$, since the ``first'' intervals of two interval pairs in $X$ are intersecting.
    \item If $P$ is not \emph{temporal}, that is, the time-arcs in $P$ are not non-decreasing, then we have a contradiction to $X$ being an independent set in $(I,c)$, since the ``second'' intervals of two interval pairs in $X$ are intersecting.
\end{itemize}
Note that by construction we also have that $P\in \mathcal{R}^\mathcal{C}$.  Furthermore, since we assume all interval pairs in $I$ have distinct endpoints, we have that for different colourful independent sets $X,X'$ in $(I,c)$ using colours from $\mathcal{C}$, we obtain different temporal paths $P,P'\in \mathcal{R}^\mathcal{C}$, that is, we have described an injective mapping from the colourful independent sets in $(I,c)$ that use colours from $\mathcal{C}$ to the temporal paths in $\mathcal{R}^\mathcal{C}$.

Overall, we now have shown that the set $\mathcal{R}^\mathcal{C}$ contains temporal $(s,z)$-paths that have a one-to-one correspondence with colourful independent sets in $(I,c)$ that use colours in $\mathcal{C}$, which implies that  $|\mathcal{R}^\mathcal{C}|=|\{X\mid X \text{ is a colourful independent set in } (I,c) \text{ and each vertex in } X \text{ has a}$ $\text{colour from }\mathcal{C}\}|$. This finishes the proof of the claim. 
\end{claimproof}

Using an oracle $A$ for \TemporalPaths, we can solve the \textsc{$\oplus$Multicoloured Independent Set on 2-Track Interval Graphs} instance $(I,c)$ as follows. We use a dynamic programming table $F : 2^{[k]} \rightarrow \{0,1\}$ where $F(\mathcal{C})$ for some $\mathcal{C}\subseteq [k]$ equals the parity of independent sets in $(I,c)$ that have exactly one vertex of each colour in $\mathcal{C}$.
\begin{align*}
    F(\emptyset) &= 0,\\
    F(\mathcal{C}) &= (A(\mathcal{G}^{\mathcal{C}})+\sum_{\mathcal{C}'\subset \mathcal{C}} F(\mathcal{C}'))\bmod 2.
\end{align*}
The correctness follows directly from \cref{claim:count}. Computing $F([k])$ requires $\mathcal{O}(2^k)$ calls to the oracle $A$ and $\mathcal{O}(4^k)\cdot |(I,c)|^{\mathcal{O}(1)}$ time overall (including computing the temporal graphs $\mathcal{G}^\mathcal{C}$).

\smallskip
\noindent\emph{Replacing directed time-arcs with undirected time-edges.} Lastly, we show how to replace the directed time-arcs in the constructed temporal graphs with undirected time-edges such that the correctness of our reduction is preserved. The main idea is to use appropriate edge subdivisions. We assume these replacements are done before scaling all time labels to obtain integer time labels.
\begin{itemize}
\item For all $x=([a,b],[a',b'])\in I$ with $c(x)\in\mathcal{C}$ and all $i\in\mathcal{C}\setminus\{c(x)\}$ we add a new vertex $u_{(x,i)}$ replace each time-arc $(v_b,w_i,b')$ with the two time-edges $(\{v_b,u_{(x,i)}\},b')$, $(\{u_{(x,i)},w_i\},b'+\varepsilon)$.
\item For all $x=([a,b],[a',b'])\in I$ we add a new vertex $u_{(x,z')}$ replace each time-arc $(v_b,z',b')$ with the two time-edges $(\{v_b,u_{(x,z')}\},b')$, $(\{u_{(x,i)},z'\},b'+\varepsilon)$.
\item We replace all other time-arcs with undirected versions of themselves.
\end{itemize}

\begin{figure}[t]
\begin{center}
\begin{tikzpicture}[line width=1pt, scale=.8, xscale=1.2]

    \node[vert,label=below:$v_1$] (V1) at (0,0) {}; 
    \node[vert,label=below:$v_2$] (V2) at (2,0) {}; 
    \node[vert,label=below:$v_3$] (V3) at (4,0) {}; 
    \node[vert,label=below:$v_4$] (V4) at (6,0) {}; 
    \node[vert,label=below:$v_5$] (V5) at (8,0) {}; 
    \node[vert,label=below:$v_6$] (V6) at (10,0) {};
    \node (A) at (11.25,0) {$\ldots$};
    \node[vert,label=below:$v_{2|I|}$] (Vi) at (12.5,0) {};
    
    \node[vert,label=left:$u_{x}$] (U1) at (0,1.5) {}; 

    \node[vert,label=left:$u_{(x,2)}$] (U11) at (4,2) {}; 
    \node[vert,label=left:$u_{(x,3)}$] (U12) at (6,2) {}; 
    \node[vert,label=left:$u_{(x,k)}$] (U13) at (8.5,2) {}; 
    \node[vert,label=right:$u_{(x,z')}$] (U14) at (10.5,2) {}; 
    
    \node[vert,label=left:$s$] (S) at (0,4) {}; 
    \node[vert,label=above:$z'$] (Z1) at (10.5,4) {};
    \node[vert,label=right:$z$] (Z2) at (12.5,4) {};
    \node[vert,label=above:$w_1$] (W1) at (2,4) {};
    \node[vert,label=above:$w_2$] (W2) at (4,4) {};
    \node[vert,label=above:$w_3$] (W3) at (6,4) {};
    \node (B) at (7.25,4) {$\ldots$};
    \node[vert,label=above:$w_k$] (Wk) at (8.5,4) {};

	\draw (S) --node[timelabel] {$1$} (W1);
	\draw (S) edge[bend left=50] node[timelabel] {$1$} (W2);
	\draw (S) edge[bend left=50] node[timelabel] {$1$} (W3);
	\draw (S) edge[bend left=50] node[timelabel] {$1$} (Wk);
	\draw (Z1) --node[timelabel] {$2|I|+1$} (Z2);

	\draw[red] (W1) --node[timelabel] {$2$} (U1);
	\draw[red] (U1) --node[timelabel] {$5,6-\varepsilon$} (V1);
	
	\draw[red] (V4) --node[timelabel] {$5$} (U14);
	\draw[red] (V4) --node[timelabel] {$5$} (U11);
	\draw[red] (V4) --node[timelabel] {$5$} (U12);
	\draw[red] (V4) --node[timelabel] {$5$} (U13);
	
	\draw[red] (U14) --node[timelabel] {$5+\varepsilon$} (Z1);
	\draw[red] (U11) --node[timelabel] {$5+\varepsilon$} (W2);
	\draw[red] (U12) --node[timelabel] {$5+\varepsilon$} (W3);
	\draw[red] (U13) --node[timelabel] {$5+\varepsilon$} (Wk);
	
	\draw[red] (V1) --node[timelabel] {$5,6-2\varepsilon$} (V2);
	\draw[red] (V2) --node[timelabel] {$5,6-3\varepsilon$} (V3);
	\draw[red] (V3) --node[timelabel] {$5,6-4\varepsilon$} (V4);
	
\end{tikzpicture}
    \end{center}
    \caption{Illustration of the modified undirected version of $\mathcal{G}^\mathcal{C}$ with $\mathcal{C}=\{1,3\}$ and one interval pair $x=([1,4],[2,5])\in I$ with $c(x)=1$. The edges added for $x$ are depicted in red.}\label{fig:hardness2}
\end{figure}

Consider the undirected temporal graph obtained from the modifications above, an illustration is given in \cref{fig:hardness2}. We can make the following observations.
\begin{itemize}
    \item Once a temporal $(s,z)$-path arrives at a vertex $w_i$ for some $i\in \mathcal{C}$, then it can only continue to a vertex $u_{x}$ for some $x\in I$ with $c(x)=i$. It cannot continue to a vertex $u_{(x,i)}$, since from there it cannot continue without revisiting vertex $w_i$. For the same reason it cannot return to vertex~$s$.
    \item Once a temporal $(s,z)$-path arrives at a vertex $u_x$ for some $x=([a,b],[a',b'])\in I$, then it can only continue to vertex $v_a$, since otherwise it would revisit vertex $w_{c(x)}$.
    \item Once a temporal $(s,z)$-path arrives at a vertex $u_{(x,i)}$ for some $x=([a,b],[a',b'])\in I$ with $c(x)\neq i$, then it can only continue to a vertex $w_i$, since otherwise it would revisit vertex $v_b$.    
    \item Once a temporal $(s,z)$-path arrives at a vertex $u_{(x,z')}$ for some $x=([a,b],[a',b'])\in I$, then it can only continue to a vertex $z'$, since otherwise it would revisit vertex $v_b$.
    \item Once a temporal $(s,z)$-path arrives at vertex $z'$, then it can only continue to vertex $z$, since otherwise it would revisit vertex $z'$.
\end{itemize}
Furthermore, we can observe that if a temporal $(s,z)$-path arrives at a vertex $v_a$ coming directly from a vertex $v_{a-1}$, it can only continue to $v_{a+1}$ or some vertex $u_{(x,i)}$ for some $x\in I$ with $c(x)\neq i$. It cannot continue to a vertex $u_{x}$ for some $x\in I$, since from there it cannot continue without revisiting vertices.

In all so far mentioned cases, the temporal $(s,z)$-paths follow edges along the directions they had in the directed temporal graph construction.

The only ``problematic'' case left is when a temporal $(s,z)$-path arrives at a vertex $v_a$ coming directly from vertex $u_x$ for $x=([a,b],[a',b'])\in I$. Then it can continue in the ``wrong direction'' to $v_{a-1}$. However, since we assume all interval pairs in $I$ have unique endpoint, the edge $\{v_{a-1},v_a\}$ does not have label $b'$ and the next larger label is at least $b'+1$, whereas the edge $\{u_x,v_a\}$ has labels $b'$ and $b'+1-a\varepsilon$. It follows that this has the same effect as ``leaving the early lane'', see proof of \cref{claim:cheating}, and hence we have that there is an even number of such temporal $(s,z)$-paths. Formally, this can be shown by an analogous proof to the one of \cref{claim:cheating}.

Lastly, it is easy to verify that the modifications do not increase the feedback vertex number of the underlying graph, because we only subdivide edges.
\end{proof}

\subsection{Approximation Hardness}\label{sec:approx-hard}

In this section, we prove the following result.
\begin{theorem}\label{cor:approxhardness}
There is no fully polynomial randomised approximation scheme (FPRAS) for \TemporalPaths unless randomised polynomial time (RP) equals NP.
\end{theorem}

It is straightforward to reduce from \Paths, the problem of counting $(s,z)$-paths in a static graph, to \TemporalPaths: we set every edge to be active at time one only.  Hardness of \Paths is proved easily by imitating the reduction used by \citet{jerrum86generation}
to demonstrate that there is no FPRAS to count directed cycles in a directed graph.\footnote{Indeed, the fact that this technique can be adapted to demonstrate the hardness of approximately counting $(s,z)$-paths is noted without proof by
\citet{sinclairThesis}. }  We note that the reduction also rules out the existence of any polynomial-time (randomised) approximation algorithm achieving any polynomial additive error.

\begin{theorem}
\label{thm:approxhardness}
There is no FPRAS for \Paths unless RP$=$NP.
\end{theorem}

\begin{proof}
We show that the existence of an FPRAS for \Paths would give a randomised polynomial-time algorithm (with one-sided error) to solve the NP-hard problem \textsc{Hamilton Path}.  We begin by observing that the existence of an FPRAS for \Paths implies the existence of an FPAUS to sample approximately uniformly at random from the set of $(s,z)$-paths.  This fact follows directly from a general result of \citet{jerrum86generation}, together with the fact that \Paths is downward self-reducible (the set of solutions can be partitioned according to the first edge of the path, with solutions including $sv$ as the first edge corresponding to $(v,z)$-paths in the graph obtained by deleting $s$).  Fron mow on, therefore, we assume the existence of an FPAUS to sample $(s,z)$-paths, and show that this gives rise to a randomised polynomial-time algorithm for \textsc{Hamilton Path}.

Let $G=(V,E)$ be an instance of \textsc{Hamilton Path} and suppose that $G$ has $n \ge 2$ vertices and $m$ edges; we will construct a collection $\{(G_{u,v},s,z):u,v \in V\}$, where $s$ and $z$ are distinguished vertices in each graph $G_{u,v}$, such that:
\begin{itemize}
    \item if $G$ is a yes-instance for \textsc{Hamilton Path}, then there is some graph $G_{u,v}$ in which at least $2/3$ of all $(s,z)$-paths contain at least $2(n+1)\lceil 2 n \log n \rceil$ edges, and
    \item if $G$ is a no-instance for \textsc{Hamilton Path}, then none of the graphs $G_{u,v}$ contains a $(s,z)$-path with more than $2n \lceil 2 n \log n \rceil$ edges.
\end{itemize}


We now describe how to construct each graph $G_{u,v}$.  As an intermediate step, we construct an auxiliary graph $G_{u,v}'$ by adding two pendant leaves, $s$ and $z$, adjacent respectively to $u$ and $v$.  Observe that there is an $(s,z)$-path with $n$ internal vertices in $G_{u,v}'$ if and only if there is a Hamilton path in $G$ that starts at $u$ and ends at $v$.  We now define the crucial gadget of our reduction, which we call an $(x,y,\ell)$-diamond gadget.  This gadget, illustrated in \cref{fig:diamond-gadget}, consists of $\ell$ copies of $C_4$, each of which has a vertex in common with the next copy; $x$ and $y$ are vertices of the two end copies.  Formally, the $(x,y,\ell)$-diamond gadget has vertices $x = z_0,\dots,z_k = y$ and $w_i^j$ for $1 \le i \le k$ and $j \in \{1,2\}$; each vertex $w_i^j$ is adjacent to $z_{i-1}$ and $z_i$.  We then obtain $G_{u,v}$ from $G'_{u,v}$ by replacing each edge $\{x,y\}$ in $G_{u,v}'$ with an $(x,y,\ell)$-diamond gadget, where $\ell := \lceil 2 n \log n \rceil$.
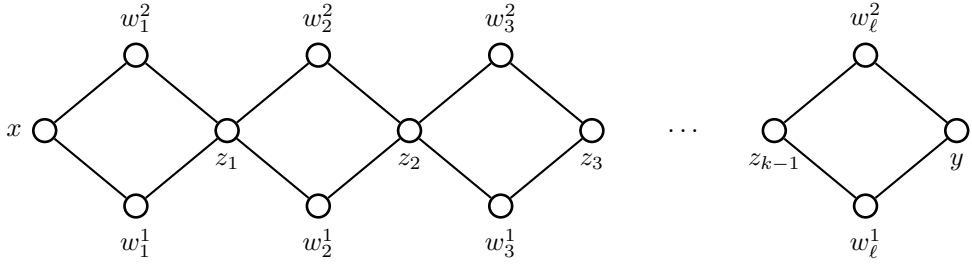
\begin{figure}[t]
\begin{center}
\begin{tikzpicture}[line width=1pt, scale=1, xscale=1.2]

    \node[vert,label=left:$x$] (x) at (0,0) {}; 
    \node[vert,label=below:$z_1$] (z1) at (2,0) {}; 
    \node[vert,label=below:$z_2$] (z2) at (4,0) {}; 
    \node[vert,label=below:$z_3$] (z3) at (6,0) {}; 
    \node (A) at (7,0) {$\ldots$};
    \node[vert,label=below:$z_{k-1}$] (zk-1) at (8,0) {}; 
    \node[vert,label=below:$y$] (y) at (10,0) {};
    \node[vert,label=below:$w_1^1$] (w11) at (1,-1) {};
    \node[vert,label=above:$w_1^2$] (w12) at (1,1) {};
    \node[vert,label=below:$w_2^1$] (w21) at (3,-1) {};
    \node[vert,label=above:$w_2^2$] (w22) at (3,1) {};
    \node[vert,label=below:$w_3^1$] (w31) at (5,-1) {};
    \node[vert,label=above:$w_3^2$] (w32) at (5,1) {};
    \node[vert,label=below:$w_\ell^1$] (wk1) at (9,-1) {};
    \node[vert,label=above:$w_\ell^2$] (wk2) at (9,1) {};
    
    \draw[edge] (x) -- (w11);
    \draw[edge] (x) -- (w12);
    \draw[edge] (z1) -- (w11);
    \draw[edge] (z1) -- (w12);
    \draw[edge] (z1) -- (w21);
    \draw[edge] (z1) -- (w22);
    \draw[edge] (z2) -- (w21);
    \draw[edge] (z2) -- (w22);
    \draw[edge] (z2) -- (w31);
    \draw[edge] (z2) -- (w32);
    \draw[edge] (z3) -- (w31);
    \draw[edge] (z3) -- (w32);
    \draw[edge] (zk-1) -- (wk1);
    \draw[edge] (zk-1) -- (wk2);
    \draw[edge] (y) -- (wk1);
    \draw[edge] (y) -- (wk2);
\end{tikzpicture}
    \end{center}
    \caption{The construction of an $(x,y,\ell)$-diamond gadget.}\label{fig:diamond-gadget}
\end{figure}

Suppose that there is an $(s,z)$-path in $G_{u,v}'$ which contains $n$ vertices of $G$.  In this case, the number of $(s,z)$-paths in $G_{u,v}$ with $2(n+1)\ell$ edges is at least $2^{(n+1)\ell} \ge 2^{2n(n+1)\log n} = n^{2n(n+1)}$, since this path corresponds to a sequence of $(n+1)\ell$ diamonds in $G_{u,v}$, each of which can be traversed in two distinct ways (via $w_i^1$ or $w_i^2$).  We now bound the number of shorter $(s,z)$-paths in $G_{u,v}$; note that any such path contains at most $2n\ell$ edges, and includes at most $n$ vertices of $V(G_{u,v}')$.  There are at most $n^n$ sequences in which vertices of $V(G_{u,v}')$ could appear on such a path, and each such sequence corresponds to $2^{n\ell}$ distinct $(s,z)$-paths (as it corresponds to a sequence of $n \ell$ diamonds).  Thus, the total number of shorter $(s,z)$-paths is at most $n^n2^{n\ell}$.  It follows that the proportion of $(s,z)$-paths in $G_{u,v}$ with at least $2(n+1)\ell$ edges is at least
\[
    \frac{2^{(n+1)\ell}}{2^{(n+1)\ell} + n^n2^{n \ell}} = \frac{2^{\ell}}{2^{\ell} + n^n} \ge \frac{n^{2n}}{n^{2n} + n^n} \ge 2/3,
\]
and hence a uniformly random path in $G_{u,v}$ has at least $2(n+1)\lceil 2n \log n \rceil$ edges with probability at least $2/3$.  Note that, if there is no $(s,z)$-path in $G_{u,v}'$ which contains $n$ vertices of $G$, then $G_{u,v}$ contains no $(s,z)$-path of length $2(n+1)\lceil 2n \log n \rceil$.

It now follows immediately that the following algorithm correctly determines the existence of a Hamilton path in $G$ with high probability.  First, construct the graphs $G_{u,v}$ as described above; it is clear that this can be done in polynomial time.  Next, apply the assumed FPAUS (with error parameter $\delta = 1/12$) to sample an approximately uniform $(s,z)$-path in each $G_{u,v}$.  If any of these sampled paths contains at least $2(n+1)\lceil 2n \log n \rceil$ edges, return YES; otherwise, return NO.  If $G$ does not contain a Hamilton path, then we are certain to return NO, whereas if there is Hamilton path in $G$ we will return YES with probability at least $7/12 > 1/2$.
\end{proof}

\section{Exact Algorithms for Temporal Path Counting}\label{sec:exactalgs}
In this section, we present several exact algorithms for \TemporalPaths. We start in \cref{sec:ptimealgs} with a polynomial-time algorithm for temporal graphs that have a forest as underlying graph.
In \cref{sec:forestgeneralization} we show that our polynomial-time algorithm can be generalised in two ways, obtaining FPT-algorithms for the so-called timed feedback vertex number and the feedback edge number of the underlying graph. In \cref{sec:treewidth} we show that \TemporalPaths is in FPT when parameterised by the treewidth of the underlying graph and the lifetime combined. 
Lastly, in \cref{sec:intervalmembership}, we give an FPT algorithm for \TemporalPaths parameterised by the so-called vertex-interval-membership-width.

\subsection{A Polynomial Time Algorithms for Forests}\label{sec:ptimealgs}
As a warm-up, we note that \TemporalPaths can be solved in polynomial time with a simple dynamic program if the underlying graph is a forest. This is used as a subroutine for algorithms presented in \cref{sec:forestgeneralization}.

\begin{theorem}
\label{thm:ptimeforest}
\TemporalPaths is solvable in $\mathcal{O}(|V|\cdot T^2)$ time if the underlying graph of the input temporal graph is a forest.
\end{theorem}
\begin{proof}
Let $\mathcal{G}=(V,\mathcal{E},T)$ together with two verices $s,z\in V$ 
be an an instance of \TemporalPaths. We argue that this instance can be solved in polynomial time if there is a unique path between~$s$ and $z$ in the underlying graph of $\mathcal{G}$. Note that this is the case if the underlying graph of $\mathcal{G}$ is a forest.

First, observe that when counting $(s,z)$-paths starting at~$s$ and arriving at~$z$, if there is a unique static path between $s$ and $z$ in the underlying graph then we need only consider time-edges between vertices of that unique static path in our temporal graph when counting, as our temporal path may not repeat vertices and so corresponds to a path in the underlying graph. Edges not lying on the unique static path between $s$ and $z$ can therefore be deleted without changing the result, so we may w.l.o.g.\ consider an instance in which the underlying graph consists only of a static path $P = (v_0, v_1, ..., v_{|P|})$ with $s = v_0$ and $z = v_{|P|}$ as the leaf vertices. 

We will base our counting on a recording at each vertex $v_i$ in $P$ of how many temporal $(s,v_i)$-paths there are starting at $s$ and arriving at $v_i$ at time $t$ or earlier. Note that there are $\mathcal{O}(|P|\cdot T) = \mathcal{O}(|V|\cdot T)$ such vertex-time pairs.

We argue by induction on $i$ that we can correctly compute this number for every vertex-time pair by dynamic programming.  As a base case, note that there is one path from $s$ to $s$ for any arrival time. 
Then we assume that we have these numbers computed correctly for some $v_{i}$ with $i\ge0$ and show how we compute them for $v_{i+1}$. Formally, our dynamic program is defined as follows.
\begin{align*}
    F(v_0=s,t) &= 1 \\
    F(v_i,t) &= \sum_{(\{v_{i-1},v_i\},t')\in\mathcal{E} \text{ with } t'\le t} F(v_{i-1},t') \text{ for } i\ge 1. 
\end{align*}
It is straightforward to check that $F(z,T)$ can be computed in the claimed running time. We now formally prove correctness by induction on $i$. That is, we prove that $F(v_i,t)$ equals the number of temporal $(s,v_i)$-paths that start at $s$ and arrive at $v_i$ at time $t$ or earlier; it will follow immediately that $F(z,T)$ is the number of $(s,z)$-paths, so it suffices to compute $F(v_i,t)$ for all $0 \le i \le |P|$.

The base case $i=0$ is trivial. Assume that $i>0$. We sum over the last time-edge of the temporal $(s,v_i)$-paths starting at $s$ and arriving at $v_i$ at time $t$ or earlier.  Let $\mathcal{P}$ be the set of all temporal $(s,v_i)$-paths starting at $s$ and arriving at $v_i$ at time $t$ or earlier that use $(\{v_{i-1},v_i\},t')\in\mathcal{E}$ as the last time-edge. All these temporal paths need to arrive at $v_{i-1}$ at time $t'$ or earlier, otherwise they cannot use time-edge $(\{v_{i-1},v_i\},t')$. Since all temporal paths in $\mathcal{P}$ do not differ in the last time-edge, the cardinality of $\mathcal{P}$ equals the number of temporal $(s,v_{i-1})$-paths starting at $s$ and arriving at $v_{i-1}$ at time $t'$ or earlier. By the induction hypothesis this number equals $F(v_{i-1},t')$.  Clearly, if the last time-edge of two temporal $(s,v_i)$-paths starting at $s$ and arriving at $v_i$ at time $t$ or earlier is different, then the two temporal paths are different, so we do not double count. 

Hence, we have shown that $F(v_i,t)$ equals the number of temporal $(s,v_i)$-paths that start at $s$ and arriving at $v_i$ at time $t$ or earlier. 
\end{proof}

\subsection{Generalisations of the Forest Algorithm}\label{sec:forestgeneralization}
In this subsection, we present two generalisations of \cref{thm:ptimeforest}. The first one results in an FPT-algorithm for the timed-feedback vertex number as a parameter and the second one in an FPT-algorithm for the feedback edge number of the underlying graph as a parameter. 
We remark that both parameters are larger than the feedback vertex number of the underlying graph, for which \cref{thm:parameterizedhardness} refutes the existence of FPT-algorithms.
Both algorithms are inspired by algorithms presented by \citet{CHMZ21} for the so-called \textsc{Restless Temporal Path} problem.

The timed feedback vertex number was introduced by \citet{CHMZ21} and, intuitively, counts the minimum number of \emph{vertex appearances} that need to be removed from a temporal graph to make its underlying graph cycle-free. Formally, it is defined as follows.

\begin{definition}[\cite{CHMZ21}]
Let $\mathcal{G}=(V,\mathcal{E},T)$ be a temporal graph. A \emph{timed feedback vertex set} of $\mathcal{G}$ is a set $X\subseteq V\times[T]$ of vertex appearances such that the underlying graph of $\mathcal{G}'=(V,\mathcal{E}',T)$ is a forest, where $\mathcal{E}':=\mathcal{E}\setminus \{(\{v,w\},t)\in \mathcal{E}\mid (v,t)\in X \vee (w,t)\in X\}$. 
The \emph{timed feedback vertex number} of a temporal graph~$\mathcal{G}$ is the minimum cardinality of a timed feedback vertex set of $\mathcal{G}$.
\end{definition}

Roughly speaking, our FPT-algorithm for the timed feedback vertex number as a parameter performs the following steps. We give a detailed proof in \cref{proof:thm:fpttfvs}.
\begin{enumerate}
    \item Compute a minimum cardinality timed feedback vertex set of the input temporal graph.
    \item Iterate over all possibilities for how a temporal path can traverse the vertex appearances in the timed feedback vertex set.
    \item For each possibility, create an instance of the so-called \textsc{\#Weighted Multicoloured Independent Set on Chordal Graphs} problem to compute the number of possibilities for connecting the vertex appearances of the timed feedback vertex set that are supposed to be traversed.
    \item Using this, compute the total number of temporal $(s,z)$-paths in the temporal input graph.
\end{enumerate}
The intuition here is that the possibilities for connecting the vertex appearances of the timed feedback vertex set that are supposed to be traversed correspond to path segments in the underlying graph of the temporal graph without the timed feedback vertex set, which is a forest. It is well-known that chordal graphs are intersection graphs of subtrees in forest \cite{GAVRIL197447}. This means that an independent set in a chordal graph corresponds to a selection of non-intersecting subtrees (which here will all be paths). The colours can be used to make sure that, for each pair of vertex appearances of the timed feedback vertex set that are supposed to be traversed directly after one another, exactly one path segment connecting them can be in the independent set. The weights can be used to model how many temporal paths follow the corresponding path segment of the underlying graph

Our algorithm follows similar ideas as the one by \citet{CHMZ21} for the \textsc{Restless Temporal Path} problem. The main difference is that we have to solve \textsc{\#Weighted Multicoloured Independent Set on Chordal Graphs} as a subroutine instead of the unweighted decision version of the problem. In the following we give a formal definition.
\problemdef{\textsc{\#Weighted Multicoloured Independent Set on Chordal Graphs}}{A chordal graph $G=(V,E)$, a colouring function $c:V\rightarrow [k]$, and a weight function $w:V\rightarrow \mathbb{N}$.}{Compute $\sum_{X\subseteq V \mid X \text{ is a multicoloured independent set in } G} \ \prod_{v\in X}w(v)$.}

We can observe that \textsc{\#Weighted Multicoloured Independent Set on Chordal Graphs} presumably cannot be solved in polynomial time. This follows directly from the NP-hardness of \textsc{Multicoloured Independent Set on Chordal Graphs}~\cite[Lemma~2]{van2015interval}.
Hence, we have the following.
\begin{observation}
\textsc{\#Weighted Multicoloured Independent Set on Chordal Graphs} cannot be solved in polynomial time unless P$=$NP.
\end{observation}

However, we can obtain an FPT-algorithm for \textsc{\#Weighted Multicoloured Independent Set on Chordal Graphs} parameterised by the number of colours. This will be sufficient for our purposes. 

To show this result, we adapt an algorithm by~\citet[Proposition~5.6]{bentert2019indu} to solve \textsc{Maximum Weight Multicoloured Independent Set on Chordal Graphs}, where given a chordal graph $G=(V,E)$, a colouring function $c:V\rightarrow [k]$, and a weight function $w:V\rightarrow \mathbb{N}$, one is asked to compute a multicoloured independent set of maximum weight in~$G$. Here, the weight of an independent set is the \emph{sum} of the weights of its vertices. Note that in our problem \textsc{\#Weighted Multicoloured Independent Set on Chordal Graphs} the weight of an independent set is the \emph{product} of the weights of its vertices.

\begin{proposition}
\label{prop:iscount}
\textsc{\#Weighted Multicoloured Independent Set on Chordal Graphs} is fixed-parameter tractable when parameterised by the number $k$ of colours.
\end{proposition}
\begin{proof}
\citet{bentert2019indu} provide a dynamic program on a tree decomposition of the input graph. Chordal graphs are known to admit a tree decomposition, where every bag is a clique, that can be computed in linear time~\cite{blair1993introduction}. Let $G=(V,E)$ be a chordal graph and $(\mathcal{B},\mathcal{T})$ a rooted tree decomposition of $G$ such that for each bag $B\in\mathcal{B}$ we have that $G[B]$ is a complete graph; note that an independent set can therefore contain at most one vertex from each bag. Furthermore, we assume that each bag is either a leaf, has one descendant, or has exactly two descendants that contain exactly the same vertices. A tree decomposition with this property can be obtained in linear time~\cite{bentert2019indu}. 
Let $R$ denote the root of $\mathcal{T}$ and for each $B\in \mathcal{B}$ let $V_B$ denote the set of all vertices in $B$ and all descendants of $B$ in $\mathcal{T}$.
Let $c:V\rightarrow [k]$ be a colouring function and $w:V\rightarrow \mathbb{N}$ be a weight function where w.l.o.g.\ we have for all $v\in V$ that $w(v)>0$ (since we may delete vertices with weight zero without changing the answer). We define the following dynamic programming table $F:\mathcal{B}\times 2^{[k]}\times \binom{V}{1}\cup\{\emptyset\}\rightarrow \mathbb{N}$. Intuitively, we want that for each $B\in\mathcal{B}$, each $C\subseteq [k]$, and each $v\in B$ the table entry $F[B,C,\{v\}]$ is the sum of weights of all independent sets $X$ in $G[V_B]$ such that $v\in X$, $|C|=|X|$, and $C=\{c(w)\mid w\in X\}$. For table entries $F[B,C,\emptyset]$ we want that the independent set does not contain any vertex of $B$.
We distinguish the following cases (where an empty sum equals 0):
\begin{itemize}
    \item Bag $B$ is a leaf in $\mathcal{T}$: For all $v\in B$ and $C\subseteq[k]$ we set $F[B,C,\{v\}]=w(v)$ if $C=\{c(v)\}$ and $F[B,C,\{v\}]=0$ otherwise. We set $F[B,C,\emptyset]=0$.
    \item Bag $B$ has one descendant $B'$ in $\mathcal{T}$: 
    For all $v\in B$ and $C\subseteq[k]$ if $c(v)\notin C$, then set $F[B,C,\{v\}]=0$. If this is not the case, then for all $v\in B\cap B'$ we set $F[B,C,\{v\}]=F[B',C,\{v\}]$, otherwise we set $F[B,C,\{v\}]=w(v)\cdot(F[B',C\setminus\{c(v)\},\emptyset] + \sum_{v'\in B'\setminus B}F[B',C\setminus\{c(v)\},\{v'\}])$. We set $F[B,C,\emptyset]=F[B',C,\emptyset]+\sum_{v'\in B'\setminus B}F[B',C,\{v'\}]$.
    \item Bag $B$ has two descendants $B_1$ and $B_2$ in $\mathcal{T}$ with $B=B_1=B_2$: We set 
    \[
    F[B,C,\{v\}]=\frac{1}{w(v)}\cdot\sum_{
    \substack{C_1,C_2\mid \\ C_1\cup C_2 \\ = C \wedge C_1\cap C_2 \\ =\{c(v)\}}}F[B_1,C_1,\{v\}]
    \cdot F[B_2,C_2,\{v\}], \text{ and}
    \]
    \[
    F[B,C,\emptyset]=\sum_{
    \substack{C_1,C_2\mid \\ 
    C_1\cup C_2
    \\ =C \wedge C_1\cap C_2 \\=\emptyset}}F[B_1,C_1,\emptyset]\cdot F[B_2,C_2,\emptyset].
    \]
\end{itemize}
The answer to our problem is $F[R,[k],\emptyset]+\sum_{v\in R} F[R,[k],\{v\}]$, where $R$ denotes the root bag of the tree decomposition.

We show the following by induction, which then implies correctness of our algorithm. Here, we call an independent set \emph{$C$-multicoloured} if it contains exactly one vertex of each colour in $C$.
\begin{equation}\label{eq1}
F[B,C,V']=\sum_{
\substack{X\subseteq V_B \mid X \text{ is a } \\ C\text{-multicoloured independent set in }\\ G \text{ such that  } X\cap B=V'}} \ \prod_{v'\in X}w(v')\tag{*}
\end{equation}
If $B$ is a leaf in $\mathcal{T}$, then \cref{eq1} is clearly fulfilled.

Assume that $B$ has one descendant $B'$ in $\mathcal{T}$. For vertices $v\in B$ that are also contained $B'$, no vertex in $B'\setminus B$ can be contained in a $C$-multicoloured independent set in $G$ that contains $v$, since $B'$ is a clique. Hence, in this case we have that setting $F[B,C,\{v\}]=F[B',C,\{v\}]$ fulfills \cref{eq1}. If a vertex $v\in B$ is not contained in $B'$, then we can add it to any independent set containing vertices from $V_{B'}\setminus B$ (noting that the properties of a tree decomposition ensure that there cannot be any edge from a vertex in $B \setminus B'$ to one in $B' \setminus B$). In order to obtain $C$-multicoloured independent sets, we sum up all weights of $(C\setminus\{c(v)\})$-multicoloured independent sets, yielding $F[B',C\setminus\{c(v)\},\emptyset] + \sum_{v'\in B'\setminus B}F[B',C\setminus\{c(v)\},\{v'\}]$. By distributivity of multiplication, we can multiply this sum with $w(v)$ to obtain the weighted sum of $C$-multicoloured independent sets that additionally contain vertex $v$. It follows that in this case, \cref{eq1} is fulfilled. The correctness for $F[B,C,\emptyset]$ is analogous.

Lastly, assume that $B$ has two descendants $B_1$ and $B_2$ in $\mathcal{T}$ with $B=B_1=B_2$. For vertices $v\in B$ with $c(v)\in C$ we sum up all possibilities of combining a $C_1$-multicoloured independent set $X_1$ in $G[V_{B_1}]$ with a $C_2$-multicoloured independent set $X_2$ in $G[V_{B_2}]$ such that we obtain a $C$-multicoloured independent set $X=X_1\cup X_2$ in $G[V_{B}]$. To this end, $C_1$ and $C_2$ must obey $C_1\cup C_2=C$ and $C_1\cap C_2=\{v\}$. The set $X$ is clearly an independent set since $V_{B_1}\cap V_{B_2}= B_1=B_2=B$ and hence all edges between vertices in $V_{B_1}$ and $V_{B_2}$ have both their endpoints in $B$ and $X$ contains exactly one vertex from $B$, namely $\{v\}=X_1\cap X_2$. Multiplying the weights the independent sets results in having the weight of $v$ appear twice in the product, hence we divide the sum of all products by $w(v)$. It follows that in this case \cref{eq1} is fulfilled for all $F[B,C,v]$. The correctness for $F[B,C,\emptyset]$ is analogous.

To obtain the claimed running time bound, note that the size of the dynamic programming table $F$ is in $\mathcal{O}(2^k\cdot |V|^2)$ and each entry can be computed in $\mathcal{O}(2^k+|V|)$ time.
\end{proof}

Using \cref{prop:iscount}, we are ready to give our FPT-algorithm for \TemporalPaths parameterised by the timed feedback vertex number.

\begin{theorem}
\label{thm:fpttfvs}
\TemporalPaths is fixed-parameter tractable when parameterised by the timed feedback vertex number of the input temporal graph.
\end{theorem}
\begin{proof}
Let $(\mathcal{G},s,z)$ be an instance of \TemporalPaths. We adapt an algorithm by~\citet[Theorem~8]{CHMZ21} for the so-called \textsc{Restless Temporal Path} problem. 

Intuitively, the algorithm iterates over all possibilities for how a temporal $(s,z)$-path can traverse the vertex appearances in the timed feedback vertex set. Then, to calculate the number of ways to select path segments to connect the timed feedback vertex set elements visited by the path, we create an instance of \textsc{\#Weighted Multicoloured Independent Set on Chordal Graphs}. Note that, informally speaking, the path segments connecting the timed feedback vertex set elements are subpaths of the forest that remains after the timed feedback vertex set is removed from the input graph, hence they form an intersection model for a chordal graph. Since a temporal path cannot revisit vertices, the timed feedback vertex set elements cannot be connected by intersecting path segments, hence we are interested in independent sets in the chordal graph representing the path segments. Furthermore, for each pair of feedback vertex set elements visited consecutively by the temporal path, we need exactly one path segment to connect them; this is modelled by giving the path segments colours.

We assume w.l.o.g.\ that in $\mathcal{G}$ there is only one time-edge incident to $s$ and that time-edge has label 1, and there is only one time-edge incident to $z$ and that time-edge has label $T$. If this is not the case, we can add two new vertices $s',z'$ to $\mathcal{G}$ and connect $s'$ to $s$ at time 1 and $z'$ to $z$ at time $T$ and the switch the names of $s,s'$ and $z,z'$, respectively.
Formally, the algorithm performs the following steps, we give a visualization in \cref{fig:tfvs}. 
\begin{enumerate}
    \item Compute a minimum timed feedback vertex set $X$ of $\mathcal{G}$. This can be done in $2^{\mathcal{O}(|X|)}\cdot|\mathcal{G}|^{\mathcal{O}(1)}$ time~\cite[Theorem~9]{CHMZ21}.
    \item Iterate over all possibilities for sets $O$, $I$ and $U$ such that $(O\cup I)\uplus U = X\cup\{(s,1),(z,T)\}$ (where $O$ and $I$ can intersect, but must both be disjoint from $U$). Slightly abusing terminology, we will refer to all triples $O$, $I$, $U$ with the mentioned property as partitions of $X\cup\{(s,1),(z,T)\}$. 
    
    Intuitively, the sets $O$, $I$, and $U$ contain elements of the timed feedback vertex set $X$ that are outgoing, incoming, or ``unused'' (that is, neither incoming nor outgoing) vertex appearances, respectively, in the temporal paths currently under considerations. The ordering $<_{O\cup I}$ specifies in which order the vertex appearances shall be visited by the temporal paths.
    \item For each partition and each ordering $<_{O\cup I}$ over $O\cup I$, create an instance of \textsc{\#Weighted Multicoloured Independent Set on Chordal Graphs} and solve it using \cref{prop:iscount}. Let $W_{<_{O\cup I}}$ denote the computed value.
    \item Output $\sum_{(O\cup I)\uplus U=X,<_{O\cup I}}W_{<_{O\cup I}}$.
\end{enumerate}

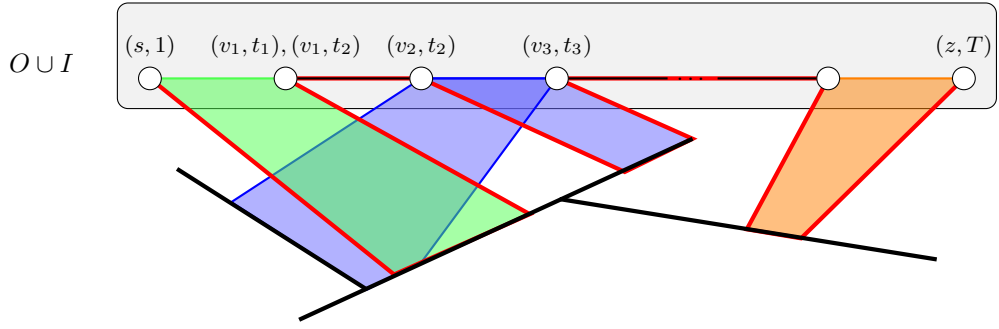
\begin{figure}[t]
\centering
	\begin{tikzpicture}[xscale=0.42,yscale=0.4,xscale=0.85]

			\draw [rounded corners,fill=gray!10] (0.8,-1) -- (0.8,2.5) -- (33.2,2.5) -- (33.2,-1) -- cycle;

			\fill [opacity=0.5,fill=blue!60] (12,0) -- (19.5,-3.1) -- (22,-2) --
			(17,0) -- cycle; \draw [blue,thick] (12,0) -- (19.5,-3.1) -- (22,-2) -- (17,0) -- cycle;

			\fill [opacity=0.5,fill=blue!60] (12,0) -- (5,-4.1) -- (10,-7) -- (12,-6.1) -- (17,0) -- cycle;
			\draw [blue,thick] (12,0) -- (5,-4.1) -- (10,-7) -- (12,-6.1) -- (17,0) -- cycle;

			\fill [opacity=0.5,fill=green!70] (2,0) -- (11,-6.5) -- (16,-4.5) -- (7,0) -- cycle;
			\draw [green!70,thick] (2,0) -- (11,-6.5) -- (16,-4.5) -- (7,0) -- cycle;

			\fill [opacity=0.5,fill=orange] (27,0) -- (24,-5) -- (26,-5.3) -- (32,0)   -- cycle;
			\draw [draw=orange,thick] (27,0) -- (24,-5) -- (26,-5.3) -- (32,0)   -- cycle;

			\draw [ultra thick,red] (2,0) -- (11,-6.5) -- (16,-4.5) -- (7,0) -- (12,0) -- (19.5,-3.1)
					-- (22,-2) -- (17,0) -- (27,0)
			 -- (24,-5) -- (26,-5.3) -- (32,0) 
					;
			\node[vert,label={\footnotesize $(s,1)$}] (s) at (2,0) {};
			\node[vert,label={\footnotesize $(v_1,t_1),(v_1,t_2)$}] (x) at (7,0) {};
			\node[vert,label={\footnotesize $(v_2,t_2)$}] (v) at (12,0) {};
			\node[vert,label={\footnotesize $(v_3,t_3)$}] (w) at (17,0) {};
			\node (dots) at (22,0) {\footnotesize $\dots$};
			\node[vert] (y) at (27,0) {};
			\node[vert,label={\footnotesize $(z,T)$}] (z) at (32,0) {};
			\draw[thick] (x) to (v);
			\draw[thick] (w) to (dots);
			\draw[thick] (y) to (dots);

			\draw[opacity=0,thick] (22,-2) to (7.5,-8);
			\draw[ultra thick] (22,-2) to (7.5,-8);
			\draw[ultra thick] (3,-3) to (10,-7);
			\draw[ultra thick] (17.1,-4) to (31,-6);

			\node at (-2,0.5) {$O\cup I$};

	\end{tikzpicture}
	\caption{Illustration of one iteration of the algorithm described  \cref{thm:fpttfvs}. The upper part shows $O\cup I$ from the current partition ordered by $<_{O\cup I}$. The lower part sketches the underlying graph of the input temporal graph $\mathcal{G}$ without the timed feedback vertex set, 
			which is a forest.
The coloured areas correspond to vertices in the constructed \textsc{\#Weighted Multicoloured Independent Set on Chordal Graphs} instance with the respective colours.
	The red thick path illustrates a temporal $(s,z)$-path which
	corresponds to a multicoloured independent set in the constructed \textsc{\#Weighted Multicoloured Independent Set on Chordal Graphs} instance.}
	\label{fig:tfvs}
\end{figure}

In the following, we describe how we construct an instance of \textsc{\#Weighted Multicoloured Independent Set on Chordal Graphs} given a partition $(O\cup I)\uplus U$ of $X\cup\{(s,1),(z,T)\}$ and an ordering $<_{O\cup I}$ over $O\cup I$. If there are two vertex appearances $(v_1,t_1),(v_2,t_2)\in O\cup I$ with $(v_1,t_1) <_{O\cup I} (v_2,t_2)$ and $t_1 >t_2$, then we set $W_{<_{O\cup I}}=0$. 
If there are two vertex appearances $(v_1,t_1),(v_2,t_2)\in O\cup I$ with $v_1=v_2$, then we set $W_{<_{O\cup I}}=0$ unless $(v_1,t_1)\in I\setminus O$ and $(v_2,t_2)\in O\setminus I$, $t_1 \le t_2$, $s\neq v_1\neq z$, and $(v_1,t_1),(v_2,t_2)$ are adjacent in $<_{O\cup I}$.
Let $(v,t)\in O\cup I$ be the smallest element in $O\cup I$ according to $<_{O\cup I}$. If $v\neq s$ or $(v,t)\in O$, then we set $W_{<_{O\cup I}}=0$. Let $(v,t)\in O\cup I$ be the largest element in $O\cup I$ according to $<_{O\cup I}$. If $v\neq z$ or $(v,t)\in I$, then we set $W_{<_{O\cup I}}=0$. Formally, whenever we set a weight $W_{<_{O\cup I}}$ to zero, we create a trivial instance of \textsc{\#Weighted Multicoloured Independent Set on Chordal Graphs} containing only one vertex with weight zero.
If none of the above is the case, let $\{(v_1,t_1), (v_2,t_2), \ldots, (v_x,t_x)\} = O\cup I$ such that $(v_i,t_i) <_{O\cup I} (v_j,t_j)$ if and only if $i<j$ and denote $V_{O\cup I}=\{v_1,\dots,v_x\}$.

Let $F$ be the underlying (static) graph of the temporal graph $\mathcal{G}'$, where $\mathcal{G}'$ is obtained from $\mathcal{G}$ by removing the timed feedback vertex set $X$. In particular, this means that $F$ is a forest. We now define a chordal graph~$G_{<_{O\cup I}}$ using an intersection model of paths in $F$. We also define a colouring function $c_{<_{O\cup I}}$ and a weight function $w_{<_{O\cup I}}$ for the vertices in $G_{<_{O\cup I}}$.
Let $(v_i,t_i), (v_{i+1},t_{i+1}) \in O\cup I$. We create the following collection $\mathcal{P}^{(i)}_{<_{O\cup I}}$ of paths in $F$.
\begin{itemize}
    \item If $(v_i,t_i)\in O$ and $(v_{i+1},t_{i+1}) \in I$, then for each $(\{v_i,w_1\},t_i)\in \mathcal{E}$ with $(w_1,t_i)\notin X$ and each $(\{v_{i+1},w_2\},t_{i+1})\in \mathcal{E}$ with $(w_2,t_{i+1})\notin X$, let $P$ be the $(w_1,w_2)$-path in $F$ (if it exists). If $P$ visits no vertex in $V_{O\cup I}$, we add $P$ to $\mathcal{P}^{(i)}_{<_{O\cup I}}$. We count the number of temporal $(w_1,w_2)$-path in $\mathcal{G}'$ which start at $w_1$ at time $t_i$ or later and arrive at $w_2$ at time $t_{i+1}$ or earlier using \cref{thm:ptimeforest} and use this number as the weight for the vertex in $G_{<_{O\cup I}}$ corresponding to $P$. If $(\{v_i,v_{i+1}\},t_i)\in \mathcal{E}$ and $t_i=t_{i+1}$, then we add an empty path to $\mathcal{P}^{(i)}_{<_{O\cup I}}$ (corresponding to an isolated vertex in $G_{<_{O\cup I}}$) with weight one for the corresponding vertex in $G_{<_{O\cup I}}$.
    \item If $(v_i,t_i)\in O$ and $(v_{i+1},t_{i+1}) \in O\setminus I$, then for each $(\{v_i,w_1\},t_i)\in \mathcal{E}$ with $(w_1,t_i)\notin X$, let $P$ be the $(w_1,v_{i+1})$-path in $F$ (if it exists). If $P$ visits no vertex in $V_{O\cup I}\setminus\{v_{i+1}\}$, we add $P$ to $\mathcal{P}^{(i)}_{<_{O\cup I}}$. We count the number of temporal $(w_1,v_{i+1})$-paths in $\mathcal{G}'$ which start at $w_1$ at time $t_i$ or later and arrive at $v_{i+1}$ at time $t_{i+1}$ or earlier using \cref{thm:ptimeforest} and use this number as the weight for the vertex in $G_{<_{O\cup I}}$ corresponding to $P$. If $(\{v_i,v_{i+1}\},t_i)\in \mathcal{E}$, then we add an empty path to $\mathcal{P}^{(i)}_{<_{O\cup I}}$ (corresponding to an isolated vertex in $G_{<_{O\cup I}}$) with weight one for the corresponding vertex in $G_{<_{O\cup I}}$.
    \item If $(v_i,t_i)\in I\setminus O$ and $(v_{i+1},t_{i+1}) \in I$, then for each $(\{v_{i+1},w_2\},t_{i+1})\in \mathcal{E}$ with $(w_2,t_{i+1})\notin X$, let $P$ be the $(v_i,w_2)$-path in $F$ (if it exists). If $P$ visits no vertex in $V_{O\cup I}\setminus\{v_{i}\}$, we add $P$ to $\mathcal{P}^{(i)}_{<_{O\cup I}}$. We count the number of temporal $(v_i,w_2)$-path in $\mathcal{G}'$ which start at $v_i$ at time $t_i$ or later and arrive at $w_2$ at time $t_{i+1}$ or earlier using \cref{thm:ptimeforest} and use this number as the weight for the vertex in $G_{<_{O\cup I}}$ corresponding to $P$. If $(\{v_i,v_{i+1}\},t_{i+1})\in \mathcal{E}$, then we add an empty path to $\mathcal{P}^{(i)}_{<_{O\cup I}}$ (corresponding to an isolated vertex in $G_{<_{O\cup I}}$) with weight one for the corresponding vertex in $G_{<_{O\cup I}}$.
    \item If $(v_i,t_i)\in I\setminus O$ and $(v_{i+1},t_{i+1}) \in O\setminus I$, then let $P$ be the $(v_i,v_{i+1})$-path in $F$ (if it exists). If $P$ visits no vertex in $V_{O\cup I}\setminus\{v_i, v_{i+1}\}$, we add $P$ to $\mathcal{P}^{(i)}_{<_{O\cup I}}$. We count the number of temporal $(v_i,v_{i+1})$-path in $\mathcal{G}'$ which start at $v_i$ at time $t_i$ or later and arrive at $v_{i+1}$ at time $t_{i+1}$ or earlier using \cref{thm:ptimeforest} and use this number as the weight for the vertex in $G_{<_{O\cup I}}$ corresponding to $P$. 
    If $v_i=v_{i+1}$, then we add an empty path to $\mathcal{P}^{(i)}_{<_{O\cup I}}$ (corresponding to an isolated vertex in $G_{<_{O\cup I}}$) with weight one for the corresponding vertex in $G_{<_{O\cup I}}$.
\end{itemize}
We give all vertices in $G_{<_{O\cup I}}$ corresponding to paths in $\mathcal{P}^{(i)}_{<_{O\cup I}}$ colour $i$. This completes the description of the \textsc{\#Weighted Multicoloured Independent Set on Chordal Graphs} instance for partition $(O\cup I)\uplus U$ of $X$ and ordering $<_{O\cup I}$ over $O\cup I$.

Note that given a partition $(O\cup I)\uplus U$ of $X\cup\{(s,1),(z,T)\}$ and an ordering $<_{O\cup I}$ over $O\cup I$, the \textsc{\#Weighted Multicoloured Independent Set on Chordal Graphs} instance can be constructed in polynomial time and that the number of colours is in $\mathcal{O}(|X|)$. By \cref{prop:iscount}, the instance can therefore be solved in FPT-time with respect to $|X|$. Since the number of instances is in $\mathcal{O}(4^{|X|}\cdot (|X|+2)!)$, we overall obtain fixed-parameter tractability for parameter $|X|$, the timed feedback vertex number.

In the remainder we prove correctness of our algorithm. 
The correctness proof is similar to the correctness proof of \cite[Theorem~8]{CHMZ21}.
We first prove that every temporal $(s,z)$-path in $\mathcal{G}$ is counted at least once. Then we prove that every temporal $(s,z)$-path in $\mathcal{G}$ is counted at most once.

Let $P$ be a temporal $(s,z)$-path in $\mathcal{G}$. For each vertex appearance in the timed feedback vertex set $X$, we have that it is incoming, outgoing, both, or neither for $P$. Hence there is exactly one partition $(O\cup I)\uplus U$ of $X\cup\{(s,1),(z,T)\}$ that correctly distributes the timed vertex appearances into the sets $O,I,U$ for $P$. Furthermore, there is exactly one ordering $<_{O\cup I}$ over $O\cup I$ that correctly reflects in which order the vertex appearances in $O\cup I$ are visited by $P$. It is straightforward to check that when constructing the \textsc{\#Weighted Multicoloured Independent Set on Chordal Graphs}, every possibility to connect two adjacent (with respect to $<_{O\cup I}$) vertex appearances in $O\cup I$ is considered and accounts to one weight unit of a vertex in the constructed chordal graph. In particular, the connections used in $P$ are considered and account for one weight unit of exactly one multicoloured independent set in the constructed chordal graph. It follows that $P$ accounts for one weight unit in the output of the constructed \textsc{\#Weighted Multicoloured Independent Set on Chordal Graphs} instances and hence is counted at least once.

Furthermore, it is easy to see that every temporal $(s,z)$-path in $\mathcal{G}$ is counted at most once. Changing the partition $(O\cup I)\uplus U$ of $X\cup\{(s,1),(z,T)\}$ or the ordering $<_{O\cup I}$ over $O\cup I$ clearly results in considering different temporal paths. For a fixed partition $(O\cup I)\uplus U$ of $X\cup\{(s,1),(z,T)\}$ and ordering $<_{O\cup I}$ over $O\cup I$, connecting two adjacent (with respect to $<_{O\cup I}$) vertex appearances in $O\cup I$ differently also clearly results in considering different temporal paths. Lastly, no non-path temporal walks (that visit vertices multiple times) are counted. This follows by the same reasoning given in the correctness proof of \cite[Theorem~8]{CHMZ21}. 
The main idea is that if a temporal $(s,z)$-\emph{walk} visits a vertex $v$ multiple times, then this vertex either is in $V_{O\cup I}$, then these temporal $(s,z)$-walks are explicitly excluded when constructing the \textsc{\#Weighted Multicoloured Independent Set on Chordal Graphs} instances. If $v$ is in $V\setminus V_{O\cup I}$, then the vertices in the constructed \textsc{\#Weighted Multicoloured Independent Set on Chordal Graphs} instance corresponding any two path segments that contain $v$ are connected by an edge. It follows that no temporal $(s,z)$-walk visiting $v$ multiple times is considered which implies that only temporal $(s,z)$-\emph{paths} are considered.
Hence, every temporal $(s,z)$-path in $\mathcal{G}$ is counted at most once and the correctness follows.
\end{proof}

Now we consider the feedback edge number of the input temporal graph as our parameter, and show the following fixed-parameter tractability result. 
It is very similar to an algorithm by~\citet{CHMZ21} for the so-called \textsc{Restless Temporal Path} problem parameterised by the feedback edge number. 
\begin{theorem}
\label{thm:fptfen}
\TemporalPaths is fixed-parameter tractable when parameterised by the feedback edge number of the underlying graph of the input temporal graph.
\end{theorem}
\begin{proof}[Proof Sketch]
Let $(\mathcal{G},s,z)$ be an instance of \TemporalPaths. We adapt an algorithm by~\citet[Theorem~7]{CHMZ21} for the so-called \textsc{Restless Temporal Path} problem. The algorithm consist of four steps (only the last step needs adaptation to our problem):
\begin{enumerate}
	\item Exhaustively remove vertices with degree $\le 1$ from the underlying graph of $\mathcal{G}$ (except $s$ and $z$). Let~$G'$ be the resulting (static) graph.
	\item Compute a minimum feedback edge set $F$ of $G'$. Let $f := |F|$.
	\item Let $V^{\ge 3}$ denote all vertices of $G'$ with degree at least three. 
		Partition the forest $G' - F$ into a set of maximal paths $\mathcal{P}$ with endpoints in $\bigcup_{e\in F} e \cup V^{\ge 3} \cup \{s,z\}$, and intermediate vertices all of degree 2.
		It holds that $|\mathcal{P}| \in \mathcal{O}(f)$~\cite[Lemma~2]{BentertDKNN20}.
	\item Any temporal $(s,z)$-path in $\mathcal{G}$ can be formed with time-edges whose underlying edges are feedback edges from $F$ or form paths from $\mathcal{P}$. 
		Enumerate all $2^{\mathcal{O}(f)}$ sequences of underlying edges that a temporal $(s,z)$-path in $\mathcal{G}$ can follow and for each one count the temporal $(s,z)$-paths following these underlying edges using \cref{thm:ptimeforest}. Add up all path counts.
\end{enumerate}
The correctness follows from the correctness of~\cite[Theorem~7]{CHMZ21} and the observation that due to the exhaustive search, all temporal $(s,z)$-paths in $\mathcal{G}$ are considered and correctly counted. 
\end{proof}

\subsection{Parameterisation by Treewidth and Lifetime}\label{sec:treewidth}

Our goal in this subsection is to demonstrate that \TemporalPaths is in FPT when parameterised simultaneously by the treewidth of the underlying graph and the lifetime; to do this we give an MSO-encoding of the problem and make use of the counting version of Courcelle's theorem for model-checking on relational structures~\cite{courcelleBook}. 

We begin by recalling some basic notation and terminology for relational structures.
A \emph{relational vocabulary} is a finite set $\tau$ of relation symbols, each of which is associated with a natural number, known as its \emph{arity}.  Given any relational vocabulary $\tau$, a \emph{$\tau$-structure} is a pair $\mathcal{A} = (A, \{R^{\mathcal{A}}\mid R \in \tau\})$; $A$ is said to be the \emph{universe} of $\mathcal{A}$ while, for each $R \in \tau$, the \emph{interpretation} $R^{\mathcal{A}}$ of $R$ in $\mathcal{A}$ is a subset of $A^r$, where $r$ is the arity of $R$.  Here we are interested only in relational structures with finite universe, and where the maximum arity of any relation is two.  

For any vocabulary $\tau$, the set of \emph{first-order formulas} is built up from a countably infinite set of individual variables $x_1,x_2,\ldots$, the relation symbols $R \in \tau$, the connectives $\wedge, \vee, \neg$ and the quantifiers $\forall x, \exists x$ ranging over elements of the universe of the structure (for notational convenience we will also use standard shorthand such as $\implies$ and $\in$).  \emph{Monadic second-order logic} additionally allows quantification over subsets of the universe via unary relation variables (which we will call \emph{set variables}); we call a formula in monadic second-order logic an \emph{MSO-formula}.  Given an MSO-formula $\psi$, an individual variable $x$ (respectively a unary relation variable $X$) appearing in $\phi$ is said to be a \emph{free variable} if $x$ (respectively $X$) is not in the scope of a quantifier $\exists x$ or $\forall x$ (respectively $\exists X$ or $\forall X$).  We write $\psi(X_1,\dots,X_j,x_1,\ldots,x_{\ell})$ for a formula $\psi$ with with free relation variables $X_1,\ldots,X_j$ and individual variables $x_1,\ldots,x_{\ell}$.  Given subsets $A_1,\dots,A_j \subseteq A$ (formally these define interpretations of unary relation variables over $A$) and elements $a_1,\dots,a_{\ell} \in A$, we write $\mathcal{A} \models \psi(A_1,\ldots,A_j,a_1,\ldots,a_{\ell})$ to mean that $\mathcal{A}$ satisfies $\psi$ if the variables $X_1,\ldots,X_j,x_1,\ldots,x_{\ell}$ are interpreted as $A_1,\ldots,A_j,a_1,\ldots,a_{\ell}$ respectively.  We further define the set of satisfying assignments of a formula $\psi$ by
\[\psi(\mathcal{A}) := \{(A_1,\ldots,A_j,a_1,\ldots,a_{\ell}): A_1,\ldots,A_j \subseteq A, a_1,\ldots,a_{\ell} \in A, \mathcal{A} \models \psi(A_1,\ldots,A_j,a_1,\ldots,a_{\ell})\}.\] 
This definition can also be extended to formulas with no free variables: in this case $\phi(\mathcal{A})$ is a set containing only the empty tuple if $\mathcal{A} \models \phi$, and the empty set otherwise.

A \emph{tree decomposition} of a $\tau$-structure $\mathcal{A}$ with universe $A$ is a pair $(T,\mathcal{B})$, where $T = (V_T,E_T)$ is a tree and $\mathcal{B} = (B_v)_{v \in V_T}$ is a collection of subsets of $A$ such that:
\begin{enumerate}
    \item for all $a \in A$, the set $\{v \in V_T: a \in B_v\}$ is non-empty and induces a connected subtree in $T$, and
    \item for every relation symbol $R \in \tau$ and every tuple $(a_1,\dots,a_r) \in R^{\mathcal{A}}$, there exists $B_v \in \mathcal{B}$ such that $a_1,\dots,a_r \in B_V$.
\end{enumerate}
As for graphs, the \emph{width} of the tree decomposition $(T,\mathcal{B})$ is $\max_{B_v \in \mathcal{B}} |B_v| - 1$, and the \emph{treewidth} of $\mathcal{A}$ is the minimum width over all tree decompositions of $\mathcal{A}$.

We can now define the monadic second-order counting problem.

\problemdef{\textsc{\#MSO}}{A relational structure $\mathcal{A}$ and an MSO-formula $\psi$.}{Compute $|\psi(\mathcal{A})|$.}

Our strategy is to demonstrate that \TemporalPaths is a special case of this general problem.  To prove our main result we will then apply the following meta-theorem, which can be deduced immediately from \cite[Theorems~6.56 and 9.21]{courcelleBook}, together with the discussion immediately after \cite[Theorem~9.21]{courcelleBook}.

\begin{theorem}[\cite{courcelleBook}]\label{thm:courcelle-count}
Let $\psi(X_1,\dots,X_j,x_1,\dots,x_{\ell})$ be an MSO-formula with free set variables $X_1,\dots,X_j$ and free individual variables $x_1,\dots,x_{\ell}$, and let $\mathcal{A}$ be a relational structure on universe $A$.  Given a width-$w$ tree decomposition of $\mathcal{A}$, the cardinality of the set $\psi(\mathcal{A})$ can be computed in time $f(w,k) \cdot \| \mathcal{A} \|$, where $k$ is the length of the formula $\psi$ and $\| \mathcal{A} \|$ denotes the size of the structure~$\mathcal{A}$.
\end{theorem}

We now have all the ingredients to prove the main result of this section. We remark that \cref{thm:parameterizedhardness} implies that we cannot hope to obtain fixed-parameter tractability by the treewidth of the underlying graph as the only parameter. The observation that \TemporalPaths is \#P-hard for lifetime one implies that we also cannot remove the treewidth from this parameterisation.

\begin{theorem}
\label{thm:treewidth}
\TemporalPaths is in FPT when parameterised by the combination of the treewidth of the underlying graph and the lifetime.
\end{theorem}
\begin{proof}
It suffices to demonstrate that, given an arbitrary instance $(\mathcal{G},s,z)$ of \TemporalPaths, where $\mathcal{G} = (V,\mathcal{E},T)$, we can efficiently construct a relational structure $\mathcal{A}$ and a collection of at most $T$ MSO-formulas $\psi_1,\dots,\psi_T$ such that the number of temporal $(s,z)$-paths in $\mathcal{G}$ that use exactly edges active at exactly $\ell$ distinct timesteps is equal to $|\psi_{\ell}(\mathcal{A})|$; it follows that the total number of temporal $(s,z)$-paths in $\mathcal{G}$ is $\sum_{1 \le \ell \le T} |\psi_{\ell}(\mathcal{A})|$.  Provided that both the treewidth of $\mathcal{A}$ and the length of each formula $\psi_{\ell}$ can be bounded by functions of the treewidth of the underlying input graph and the lifetime of $\mathcal{G}$, the result will then follow immediately from \cref{thm:courcelle-count}.

We begin by defining the relational structure $\mathcal{A}$ that will encode our instance of \TemporalPaths.  Suppose that $\mathcal{E} = (E_1,\dots,E_T)$ and the underlying graph of $\mathcal{G}$ is $G = (V,E)$ (so $E = \bigcup_{i \in [T]} E_i$).  The universe of $\mathcal{A}$ is
\[
    A = V \cup E \cup [T].
\]
The structure $\mathcal{A}$ has four relation symbols, with the following interpretations:
\begin{itemize}
    \item $\appears$: for $e \in E$ and $i \in [T]$, we have $\appears(e,i)$ if and only if $e \in E_i$;
    \item $\inc$: for $v \in V$ and $e \in E$, we have $\inc(v,e)$ if and only if $v$ is an endpoint of $e$;
    \item $\equal$: for $x,y \in V \cup E$, we have $\equal(x,y)$ if and only if $x$ and $y$ are the same element of the universe;
    \item $\lessthan$: for $t_1,t_2 \in [T]$, we have $\lessthan(t_1,t_2)$ if and only if $t_1 \le t_2$.
\end{itemize}

We now bound the treewidth of $\mathcal{A}$.  Let $(T,\mathcal{B})$ be a tree decomposition for $G$ of width $w$; we will describe a strategy for constructing a tree decomposition $(T,\mathcal{B}')$ for $\mathcal{A}$, indexed by the same tree $T = (V_T,E_T)$.  Fix a vertex $v \in V_T$, and let $B_v$ be the corresponding element of $\mathcal{B}$.  We define the corresponding element of $\mathcal{B}'$ to be 
\[
    B_v' := B_v \cup \{e=vw \in E: v,w \in B_v\} \cup [T].
\]
It is straightforwad to verify that, with this definition, $(T,\mathcal{B}')$ is indeed a tree decomposition for $\mathcal{A}$; moreover, it is immediate that $|B_v'| \le |B_v| + \binom{|B_v|}{2} + T \le w + 1 + \binom{w+1}{2} + T < (w+1)^2 + T$ and hence that the treewidth of $\mathcal{A}$ is at most $(w+1)^2 + T$.  

We now proceed to define the formula $\psi$.  We begin by introducing several subformulas.  We first define two formulas which encode the fact that a vertex is incident with exactly one or two edges from a given set respectively:
\begin{align*}
    \degone(v,E') := & \Big(\exists e \big(e \in E' \wedge \inc(v,e)\big)\Big)\\
    & \qquad \wedge \Big(\forall e_1 \forall e_2 \big((e_1 \in E' \wedge e_2 \in E' \wedge \inc(v,e_1) \wedge \inc(v,e_2) \implies \equal(e_1,e_2) \big)\Big)
\end{align*}
is true if and only if vertex $v$ is incident with exactly one edge in $E'$, whereas
\begin{align*}
    \degtwo(v,E') := & \bigg(\exists e_1 \exists e_2  \Big(e_1 \in E' \wedge e_2 \in E' \wedge \inc(v,e_1) \wedge \inc(v,e_2) \wedge \neg \big(\equal(e_1,e_2)\big)\Big)\bigg) \\
    & \qquad \wedge \bigg(\forall e_3 \forall e_4 \forall e_5 \Big( \big(e_3 \in E' \wedge e_4 \in E' \wedge e_5 \in E' \wedge \inc(v,e_3) \wedge \inc(v,e_4) \wedge \inc(v,e_5) \big) \\
    & \qquad \qquad \implies \big(\equal(e_3,e_4) \vee \equal(e_4,e_5) \vee \equal(e_3,e_5) \big) \Big) \bigg)
\end{align*}
is true if and only if vertex $v$ is incident with exactly two edges in $E'$.  Our next subformula is true if and only if the edges in the set $E'$ form a connected subgraph:
\begin{align*}
    \conn(E') &:= \forall F \bigg(\Big(\exists e_1 \exists e_2 \big(e_1 \in E' \wedge e_1 \in F \wedge e_2 \in E' \wedge \neg(e_2 \in F) \big)\Big) \\
    & \qquad \implies \Big(\exists e_3 \exists e_4 \exists v \big( e_3 \in E' \wedge e_3 \in F \wedge e_4 \in E' \wedge \neg(e_4 \in F) \wedge \inc(v,e_3) \wedge \inc(v,e_4) \big)\Big)\bigg).
\end{align*}
For our final subformula, we use $\degone$, $\degtwo$ and $\conn$ to define a formula that is true if and only if the set $E'$ of edges forms a path in $G$ with endpoints $x$ and $y$:
\begin{align*}
    \ispath(x,y,E') := & \conn(E') \wedge \degone(x,E') \wedge \degone(y,E') \\
    & \qquad \wedge \Bigg(\forall v \in V \Big(\neg\big(\equal(v,x) \vee \equal(v,y)\big) \\ & \qquad \qquad \implies \Big(\big(\forall e \in E' \neg \inc(v,e)\big) \vee \degtwo(v,E')\Big)\bigg)\Bigg).     
\end{align*}
We can now define the formula $\psi_{\ell}$ as
\begin{align*}
    \psi_{\ell}(E_0,\dots,E_{\ell},t_0,\dots,t_{\ell}) := & \\
    \exists v_0 \dots \exists v_{\ell+1} 
    \Bigg( &(v_0 = s) \wedge (v_{\ell + 1} = z) \wedge \bigwedge_{0 \le i \le \ell} \ispath(v_i,v_{i+1},E_i)\\ 
    & \wedge \bigwedge_{0 \le i \le \ell - 1} \lessthan(t_i,t_{i+1}) \wedge \bigwedge_{0 \le i \le \ell} \bigwedge_{e \in E_i} \appears(e,t_i) \\
    & \wedge \exists E' \bigg( \Big(e \in E' \Longleftrightarrow \bigvee_{0 \le i \le \ell} e \in E_i\Big) \wedge \ispath(s,z,E') \bigg) \Bigg).    
\end{align*}
In this formula, the times $t_0,\dots,t_{\ell}$ are the times at which at least one edge in the temporal path is active; each set $E_i$ is the set of edges in the temporal walk active at time $t_i$, and $v_i$ is the vertex at which the path segment active at time $t_i$ joins the path segment active at time $t_{i+1}$.  The first part of the formula,
\[
(v_0 = s) \wedge (v_{\ell + 1} = z) \wedge \bigwedge_{0 \le i \le \ell} \ispath(v_i,v_{i+1},E_i),
\]
ensures that the first and last vertices on our path are $s$ and $z$ respectively, and that the edges in each set $E_i$ do indeed form a path in the underlying graph from $v_i$ to $v_{i+1}$ for each $i$.  The second part of the formula enforces temporal constraints: we verify that the times at which consecutive path segments are traversed are weakly increasing, and that every edge in $E_i$ is active at the appropriate time $t_i$.

The parts of the formula described so far describe a temporal walk from $s$ to $z$ (with appropriate departure and arrival times), but it is possible that this walk revisits vertices.  The purpose of the final part of the formula, 
\[
\exists E' \bigg( \Big(e \in E' \Longleftrightarrow \bigvee_{0 \le i \le \ell} e \in E_i\Big) \wedge \ispath(s,z,E') \bigg),
\]
is to avoid this: $E'$ is the union of all edges used in the temporal walk, and this part of the formula verifies that this set of edges does indeed form a path from $s$ to $z$ in the underlying graph.

It is clear, therefore, that $(E_0,\dots,E_{\ell},t_0,\dots,t_{\ell}) \in \psi_{\ell}(\mathcal{A})$ if and only if there is a temporal path from $s$ to $z$ in $\mathcal{G}$ in which, for $0 \le i \le \ell$, the edges of $E_i$ are traversed at time $t_i$.  Moreover, there is a one-to-one correspondence between tuples $(E_0,\dots,E_{\ell},t_0,\dots,t_{\ell})$ and $(s,z)$-paths in which each set of edges $E_i$ is traversed at time $t_i$.  It follows that $|\psi_{\ell}(\mathcal{A})|$ is the number of temporal $(s,z)$-paths in $\mathcal{G}$ which use edges active at precisely $\ell$ distinct timesteps.  Summing over all permitted choices of $\ell$, we see that the number of temporal $(s,z)$-paths in $\mathcal{G}$ that start at time $t_s$ or later and arrive at time $t_z$ or earlier is precisely equal to
\[
    \sum_{0 \le \ell \le t_z - t_s} |\psi_{\ell}(\mathcal{A})|,
\]
as required.  It remains only to bound the length of each formula $\psi_{\ell}$.  Note first that each of the subformulas $\degone$, $\degtwo$, $\conn$ and $\ispath$ has constant length.  It follows that the length of $\psi_{\ell}$ is $\mathcal{O}(\ell) = \mathcal{O}(t_z - t_s) = \mathcal{O}(T)$, as required.
\end{proof}

\subsection{Parameterisation by Vertex-Interval-Membership-Width}\label{sec:intervalmembership}

In this subsection, we present an FPT algorithm for \TemporalPaths parameterised by the so-called vertex-interval-membership-width of the input temporal graph. The vertex-interval-membership-width is a temporal graph parameter recently introduced by \citet{BumpusM21} which, like the timed feedback vertex number, depends not only on the structure of the underlying graph but also on the assignment of times to edges.  Intuitively, the vertex-interval-membership-width counts the maximum number of vertices that are ``relevant'' at any timestep, where a vertex is considered relevant if it has an incident edge both (weakly) before and after the current timestep (so, for example, a vertex $v$ is relevant only at times when a temporal path could have entered but not yet left $v$). We remark that the vertex-interval-membership-width is unrelated to the feedback vertex number of the underlying graph.

\begin{definition}[\cite{BumpusM21}]
The \emph{vertex interval membership sequence} of a temporal graph $(G,\mathcal{E},T)$ is the sequence $(F_t)_{t \in [T]}$ of vertex-subsets of $G$ where 
\[
    F_t := \{v \in V(G)\mid \exists i \le t \le j \text{ and } u,w \in V(G) \text{ such that } \{u,v\} \in E_i \text{ and } \{w,v\} \in E_j \}.
\]
Note that we allow $u=w$.  The \emph{vertex-interval-membership-width} of $(G,\mathcal{E},T)$ is the integer $\vimw(G,\mathcal{E},T) := \max_{t \in [T]} |F_t|$.  
\end{definition}

Note that every vertex incident with an edge in $E_i$ must belong to $F_i$, and so $|E_i| \le \binom{|F_i|}{2} \le |F_i|^2$.  The vertex interval membership sequence gives us a structure we can use for dynamic programming, which we exploit to obtain the following result.

\begin{theorem}
\label{thm:vertexintervalmembershipwidth}
\TemporalPaths can be solved in time $\mathcal{O}(w^{2w^2 + w}\cdot T)$ where $T$ and $w$ are the lifetime and vertex-interval-membership-width respectively of the input graph.
\end{theorem}

In our dynamic programming algorithm, a state of the bag $F_t$ is a pair $(v,X)$, where $v \in F_t$ and $X \subseteq F_t \setminus \{v\}$.  For any state $(v,X)$ of $F_t$, we compute the number $P_t(v,X)$ of temporal paths $Q$ from $s$ to $v$, arriving by time $t$, such that $V(Q) \cap (F_t \setminus \{v\}) = X$.  Computing all such values $P_t(v,X)$ is clearly sufficient, since the total number of temporal $(s,z)$-paths is $\sum_{Y \subseteq F_T \setminus \{z\}} P_T(z,Y)$.  We compute the values for each bag $F_t$ in turn, assuming for $t \ge 1$ that we have already computed all counts corresponding to $F_{t-1}$.

\begin{proof}
Note first that we may assume w.l.o.g.\ that $z$ is incident with at least one edge in $E_T$: if not, we can discard all edges in $E_T$ and decrease the lifetime by one without changing the number of $(s,z)$-paths.  It follows that $z \in F_T$.

We proceed by dynamic programming.  For each $t \in [T]$, define the set of states of $F_t$ to be 
\[
    \mathcal{S}_t = \{(v,X): v \in F_t, X \subseteq F_t \setminus \{v\} \}.
\]
For any state $(v,X) \in \mathcal{S}_t$, we define $P_t(v,X)$ to be the number of temporal paths $Q$ from $s$ to $v$, arriving by time $t$, such that $V(Q) \cap (F_t \setminus \{v\}) = X$; for notational convenience, we adopt the convention that $P_t(v,X) = 0$ if $(v,X) \notin \mathcal{S}_t$.  It is clear from this definition that the total number of paths from $s$ to $z$ is
\[
    \sum_{Y \subseteq F_T \setminus \{z\}} P_T(z,Y).
\]
It therefore suffices to compute $P_T(x,V)$ for every state $(x,V) \in \mathcal{S}_T$.  We will in fact compute the path counts for every state of each bag $F_t$ in turn, assuming that we have already computed these counts for $F_{t-1}$.

For $F_1$, given $(v,X) \in \mathcal{S}_1$, the value of $P_1(v,X)$ is the number of paths from $s$ to $v$ using only edges of $E_1$ and precisely the vertices $X$ in $F_1$.  By definition of the vertex interval membership sequence, it is clear that every endpoint of an edge in $E_1$ must belong to $F_1$, so in fact $P_1(v,X)$ counts only $(s,v)$-paths whose vertices are precisely $X \cup \{v\}$ and whose edges belong to $E_1$.  We can consider all possibilities for such a path, and hence compute $P_1(v,X)$ in time $(|X| + 1)! = \mathcal{O}(w^w)$.

Continuing inductively, suppose we have access to $P_{i-1}(v,X)$ for all $(v,X) \in S_{i-1}$.  Fix $(v',X') \in \mathcal{S}_i$.  We claim that
\begin{equation}\label{eqn:vimw-induction}
    P_i(v',X') = \sum_{Y \cap F_i = X'}P_{i-1}(v',Y) + \sum_{\substack{Q \text{ is a path from $u$ to $v$ in } (V,E_i) \\ (Y \cup V(Q))\cap F_i = X'}} P_{i-1}(u,Y).    
\end{equation}
Here, the first term counts the temporal paths from $s$ to $v$ which arrive at time at most $i-1$, and the second counts the temporal paths from $s$ to $v$ arriving at time exactly $i$.  

To see that the first term is correct, note that any vertex in $F_i$ which belongs to a path $P$ arriving at $v'$ by time $i-1$ must also belong to $F_{i-1}$ (since it must be incident with at least one edge active at time at most $i-1$), so every vertex of $(F_i \setminus \{v'\}) \cap V(P)$ must belong to $F_{i-1}$; the path $P$ may contain additional vertices of $F_{i-1} \setminus F_i$.  It follows that the number of temporal $(s,v')$-paths arriving at $v'$ by time $i-1$ and whose intersection with $F_i$ is precisely $X'$ is exactly $\sum_{Y \cap F_i = X'}P_{i-1}(v',Y)$.

For the second term, note that any temporal path from $s$ to $v$ arriving at time exactly $i$ consists of an initial segment that reaches some vertex $u$ at time at most $i$ followed by a terminal segment which must be a path from $u$ to $v$ consisting of edges in $E_i$.  Moreover, observe that every vertex of $F_i$ that is visited on such a path must either be visited by the terminal segment of edges active at time $i$ or must belong to $F_{i-1}$: if it is not visited on the terminal segment, it must be incident with at least one edge active at a time less than or equal to $i-1$, and by assumption (by membership of $F_i$) it is also incident with at least one edge active at a time greater than or equal to $i$, and hence it must belong to $F_{i-1} \cap F_i$.  Finally, we observe that there is a one-to-one correspondence between pairs consisting of an initial and terminal path segment meeting the aforementioned conditions and temporal paths from $s$ to $v$ arriving at time exactly $i$.  Correctness of \eqref{eqn:vimw-induction} now follows immediately.

It remains to bound the time required to compute the path counts for every state of each bag.  Note first that $|\mathcal{S}_i| \le w2^{w-1}$ for each $i$.  Recall that for $\mathcal{S}_1$, we can compute $P_1(v,X)$ for each state $(v,X) \in \mathcal{S}_1$ in time $\mathcal{O}(w^w)$.  
For $i > 1$, we consider the time needed to compute the two terms of \eqref{eqn:vimw-induction} separately.  For the first term, we sum at most $2^w$ previously computed values, requiring time $\mathcal{O}(2^w)$.  For the second term, we need to find all paths $Q$ that use only vertices in $F_i$ (since all endpoints of edges in $E_i$ belong to $F_i$ by definition) and have $v$ as an endpoint.  Since $|F_i| \le w$, we find all such paths in time $\mathcal{O}(w!) = \mathcal{O}(w^w)$.  Given each such path, the corresponding term in the sum can be computed in constant time.  Overall, therefore, the time to compute $P(v,X)$ is again $\mathcal{O}(w^w)$.  Summing over all states in $\mathcal{S}_i$, we see that the time required to compute the path counts for all elements of $\mathcal{S}_i$ is $\mathcal{O}(w^{w + 1}\cdot 2^{w-1}) = \mathcal{O}(w^{2w})$.  Summing over all sets $F_t$, we obtain an overall running time of $\mathcal{O}(w^{2w}\cdot T)$, as required.
\end{proof}

\section{Approximation Algorithms for Temporal Path Counting}\label{sec:approx}

In this section we consider the problems of approximating \TemporalPaths and approximating the temporal betweenness centrality.  For \TemporalPaths, recall from \cref{sec:approx-hard} that there is unlikely to be an FPRAS for \TemporalPaths in general; in \cref{sec:approx-short}, we show that there is however an FPTRAS for \TemporalPaths when the maximum permitted path length is taken as the parameter. This in turn implies the existence of an FPTRAS for \TemporalPaths when restrictions are placed on the structure of the underlying graph that limit the length of the longest path.
We remark that \cref{thm:parameterizedhardness} and \cref{cor:approxhardness} do not rule out exact FPT-algorithms for these parameterisations. We leave open whether stronger hardness results or exact algorithms for this case can be obtained.

In \cref{sec:approx-between} we apply this approximation result to the problem of approximating temporal betweenness: we demonstrate that, whenever we can efficiently approximate \TemporalPaths, we can efficiently estimate the maximum temporal betweenness centrality over all vertices of the input graph.

\subsection{Approximately Counting Short Temporal Paths}\label{sec:approx-short}

In this subsection we consider the complexity of approximately counting $(s,z)$-paths parameterised by the length of the path.

\problemdef{\STemporalPaths}{A temporal graph $\mathcal{G}=(V,\mathcal{E},T)$, two vertices $s,z \in V$, and an integer~$k$.}{Count the temporal $(s,z)$-paths in $\mathcal{G}$ that contain exactly $k$ edges.}

We prove the following result.

\begin{theorem}\label{thm:approx-short-paths}
There is a randomised algorithm which, given as input an instance $(\mathcal{G},s,z)$ of \STemporalPaths together with error parameters $\varepsilon > 0$ and $0 < \delta < 1$, outputs an estimate $\hat{N}$ of the number of temporal $(s,z)$-paths in $\mathcal{G}$ containing exactly $k$ edges; with probability at least $1 - \delta$, $\hat{N}$ is an $\varepsilon$-approximation to the number of $(s,z)$-paths in $\mathcal{G}$ containing exactly $k$ edges.  The running time of the algorithm is $\mathcal{O}(k! e^k \log(1/\delta) \varepsilon^{-2} n^2 T^2)$.
\end{theorem}

The key ingredient in the proof is an efficient algorithm for the \emph{multicoloured} version of this problem, in which the input graph is equipped with a vertex-colouring (not necessarily proper) and we wish to count paths containing exactly one vertex of each colour.

\problemdef{\MTemporalPaths}{A temporal graph $\mathcal{G}=(V,\mathcal{E},T)$, two vertices $s,z \in V$, and a partition of $V \setminus \{s,z\}$ into colour sets $V_1 \uplus \dots \uplus V_{\ell}$.}{Count the number of temporal $(s,z)$-paths that contain exactly one vertex from each colour-set $V_1,\dots,V_{\ell}$.}

\begin{lemma}
\label{lma:partitioned-path}
\MTemporalPaths is solvable in time $\mathcal{O}((\ell+1)!n^2T^2)$.
\end{lemma}
\begin{proof}
Given a permutation $\pi \colon [\ell] \rightarrow [\ell]$, let $\paths(\pi)$ be the number of temporal $(s,z)$-paths on vertices $s,v_1,\dots,v_{\ell},z$ such that $v_i \in V_{\pi(i)}$ for each $i$.  It is clear that the total number of multicolour $(s,z)$-paths is $\sum_{\pi} \paths(\pi)$. Intuitively, this permutation orders the parts and we consider paths that traverse the parts in order.  Since the number of such permutations is $\ell!$, it suffices to demonstrate that we can compute $\paths(\pi)$ for any fixed permutation $\pi$ in time $\mathcal{O}(\ell n^2T^2)$.

To compute $\paths(\pi)$, we use a dynamic programming strategy.  For a vertex $v \in V_{\pi(i)}$ and a time $t \in [T]$, we define $\completions_i(v,t)$ to be the number of $(v,z)$-paths on vertices $v,w_{i+1},\dots,w_{\ell},z$, starting at or later than time $t$, such that $w_j \in V_{\pi(j)}$ for $i+1 \le j \le \ell$.  We claim that it suffices to compute $\completions_1(v,t)$ for each vertex $v \in V_{\pi(1)}$ and $t \in [T]$ in time $\mathcal{O}(\ell n^2T^2)$.  To see this, observe that every $(s,z)$-path respecting the partition and permutation consists of an edge from $s$ to some vertex $v \in V_{\pi(1)}$, followed by a temporal path counted in $\completions_1(v,t)$.  Specifically,
\[
    \paths(\pi) = \sum_{t \in [T]} \sum_{\substack{v \in V_{\pi(1)} \\ sv \in E_t}} \completions_1(v,t).
\]
To compute $\paths(\pi)$ it therefore suffices to sum at most $Tn$ values of $\completions_1(v,t)$, after checking the existence of a single edge corresponding to each value.

We will in fact calculate $\completions_i(u,t)$ for all $u \in V_{\pi(i)}$ and $t \in [T]$ for each $i$ in turn, starting with $\ell$.  It is easy to verify that
\[
    \completions_{\ell}(u,t) = |\{t' \in \{t,\dots,T\}: uz \in E_{t'}\}|,
\]
and hence we can compute $\completions_{\ell}(u,t)$ for all pairs $(u,t)$ in time $\mathcal{O}(\sum_{t \in [T]}|E_t|)$.  Suppose now that we have computed all values $\completions_{i+1}(u,t)$ for $u \in V_{\pi(i+1)}$ and $t \in [T]$; we explain how to compute $\completions_i(w,t')$ for $w \in V_{\pi(i)}$ and $t' \in [T]$.  Again, it is clear that
\[
    \completions_{i}(w,t') = \sum_{t' \le r \le T} \sum_{\substack{u \in V_{\pi(i+1)}\\ wu \in E_r}} \completions_{i+1}(u,r).
\]
Hence, given all values $\completions_{i+1}(u,t)$, we can compute $\completions_i(w,t')$ for any given pair $(w,t')$ in time $\mathcal{O}(Tn)$, and for all such pairs in time $\mathcal{O}(T^2n^2)$.  The total time to compute $\completions_i(u,t)$ for all $1 \le i \le \ell$, $u \in V_{\pi(i)}$ and $t \in T$ is therefore 
\[
    \mathcal{O}\left(\sum_{t' \in [T]}|E_{t'}| + (\ell - 1)n^2T^2\right) = \mathcal{O}\left(\ell n^2T^2\right),
\]
as required.
\end{proof}

Equipped with this algorithm for \MTemporalPaths, we use a standard colour-coding technique to obtain an FPTRAS for \STemporalPaths.  This involves repeatedly generating random colourings (not necessarily proper) of the vertices of $V \setminus \{s,z\}$ using $k-1$ colours; note that a single colouring can clearly be generated in time $\mathcal{O}(nk)$.  For each colouring, we solve the corresponding instance of \MTemporalPaths using the algorithm of \cref{lma:partitioned-path}.  Setting $N$ to be the sum of counts over all colourings, we return $Nk^k/k!$.  Following the argument by \citet[Section 2.1]{alon-motif}, we see that the number of colourings we must generate to obtain an $\varepsilon$-approximation to \STemporalPaths with probability at least $1 - \delta$ is $\mathcal{O}(e^k \log (1/\delta) \varepsilon^{-2})$, giving the result.

Since the maximum possible path length is bounded by a function of either the vertex cover number or the treedepth\footnote{We refer to the book of \citet{sparsitybook} for the definition of treedepth, and a proof that the maximum length of a path in a graph is bounded by a function of its treedepth.} of the underlying input graph, we immediately obtain the following corollary to \cref{thm:approx-short-paths}.

\begin{corollary}\label{cor:paths-short}
\TemporalPaths admits an FPTRAS parameterised by either vertex cover number or treedepth of the underlying input graph.
\end{corollary}

\subsection{Approximating Temporal Betweenness}\label{sec:approx-between}

Observe that it is not clear how to use an approximation algorithm for \TemporalPaths to approximate the temporal betweenness centrality for every vertex in the input graph (we give a detailed discussion in \cref{sec:countvsbetweenness}).  In this section, we address the simpler problem of determining the maximal temporal betweenness centrality of any vertex in the graph: we show that we can efficiently approximate this quantity whenever there is an FPRAS (or FPTRAS) for \TemporalPaths.

\begin{theorem}
\label{thm:between-approx}
Let $\mathcal{C}$ be a class of temporal graphs on which \TemporalPaths admits an FPRAS.  Then $\mathcal{C}$ admits an FPRAS to estimate, given an input temporal graph $\mathcal{G} = (V,\mathcal{E},T) \in \mathcal{C}$, 
$\max_{v \in V} C_B^{(\star)}(v)$,
for $\star \in \{ \text{fastest, foremost}\}$. 
Similarly, if $\mathcal{C}$ is a class of graphs on which there exists an FPTRAS for \TemporalPaths with respect to some parameterisation $\kappa$ then, with respect to the same parameterisation, $\mathcal{C}$ admits an FPTRAS to estimate, given an input temporal graph $\mathcal{G} = (V,\mathcal{E},T) \in \mathcal{C}$, 
$\max_{v \in V} C_B^{(\star)}(v)$, 
for $\star \in \{ \text{fastest, foremost}\}$.
\end{theorem}

The proof relies on the fact that we may assume that at least one vertex has temporal betweenness centrality at least $\frac{1}{n(T+1)}$, where $n$ is the number of vertices; we begin by arguing that we can efficiently identify the inputs for which this lower bound does not hold, and that in these cases the correct answer is in fact $0$.  Using this assumption, we show that the following procedure is likely to produce a good approximation to $\max_{v \in V} C_B^{(\star)}(v)$: for each vertex pair $(s,z)$, sample a large (polynomial) number of $\star$-optimal temporal $(s,z)$-paths, and record the number that contain each vertex $v$ as an internal vertex; after considering all pairs $(s,z)$, we assume that the vertex $v_{\max}$ we have seen most frequently has the maximum betweenness centrality, and return as our estimate the proportion of sampled paths that contain $v_{\max}$.  We note that, applying a general result of \citet{jerrum86generation}, we can assume the existence of an efficient algorithm to sample $\star$-optimal temporal $(s,z)$-paths almost uniformly whenever there is an FPRAS (or FPTRAS).

For the proof of \cref{thm:between-approx}, we need not only to be able to approximately count $\star$-optimal temporal $(s,z)$-paths, but to sample approximately uniformly at random from the set of all $\star$-optimal temporal $(s,z)$-paths in the input graph.  To see that we can do this, we make use of a general result of 
\citet{jerrum86generation}:
for problems that are \emph{downward self-reducible} (see \cite{jerrum86generation} for the formal definition of this property), the existence of an FPRAS implies the existence of an FPAUS (and vice versa).  We note that an analogous argument implies the same relationship for the existence of an FPTRAS and FPTAUS.  It is easy to verify that the problem \TemporalPaths has this crucial property of downward self-reducibility, since the solution space can be partitioned according to the first time-edge on the path (with solutions including a fixed first time-edge $(\{s,v\},t)$ corresponding to temporal $(v,z)$-paths in the temporal graph obtained by deleting $s$ and all time-edges appearing earlier than $t$).  It follows (using the reasoning in \cref{sec:countvsbetweenness}) that, in this setting, we also have an FPAUS (or FPTAUS) to sample fastest or foremost temporal $(s,z)$-paths.  

\begin{lemma}\label{lma:count-to-sample}
Let $\mathcal{C}$ be a class of temporal graphs on which \TemporalPaths admits an FPRAS.  Then $\mathcal{C}$ also admits FPAUS to sample almost uniformly from the set of all foremost or fastest temporal $(s,z)$-paths.  Similarly, if $\mathcal{C}$ admits an FPTRAS with respect to some parameterisation $\kappa$, then, with respect to the same parameterisation, $\mathcal{C}$ also admits an FPTAUS to sample almost uniformly from the set of all foremost or fastest temporal $(s,z)$-paths.
\end{lemma}

We will also need two standard Chernoff bounds.  

\begin{lemma}[{\cite[Theorem~4.4 and Corollary~4.6]{ProbComp}}]\label{lem:chernoff}
	Suppose $X$ is a binomial random variable with mean $\mu$. Then:
	\begin{enumerate}
	    \item for all $0 < \varepsilon \le 1$, 
	\[ \mathbb{P}(|X - \mu| \ge \varepsilon\mu) \le 2e^{-\varepsilon^2\mu/3};\]
	    \item for all $R \ge 6 \mu$, 
	    \[ \mathbb{P}(X \ge R) \le 2^{-R}. \]
	\end{enumerate}
\end{lemma}

We now have all the ingredients to prove \cref{thm:between-approx}.

\begin{proof}[Proof of \cref{thm:between-approx}]
We only prove the result for an FPRAS, as the argument for the FPTRAS is identical (the only change is that we additionally allow the running times of all algorithms to be of the form $f(k)n^{\mathcal{O}(1)}$).  We fix an optimality measure $\star \in \{\text{fastest, foremost}\}$ on which the temporal betweenness will be based. 

Assume that $V = \{v_1,\dots,v_n\}$, and fix values $\varepsilon,\delta > 0$; we may assume w.l.o.g.\ that our error parameter $\varepsilon$ satisfies $\varepsilon < 1$, and that $n > 12$.  We shall describe a randomised algorithm which, with probability at least $2/3$, returns an $\varepsilon$-approximation to $b^*:= \max_{v \in V} C_B^{(\star)}(v)$.  Using standard probability amplification techniques (running the procedure  $\mathcal{O}(\log(\delta^{-1}))$ times and returning the median output) we can increase the success probability to $1 - \delta$ without violating the running time bounds.

We begin by describing one special case which we handle differently; this will allow us to assume a lower bound on $\max_{v \in V} C_B^{(\star)}(v)$ for the rest of the proof.  Our algorithm starts by considering each vertex pair $(s,z)$ in turn, and determining in polynomial time whether there is any $\star$-optimal temporal $(s,z)$-path that contains at least one internal vertex: we do this by first using a polynomial-time algorithm to find a single $\star$-optimal temporal $(s,z)$-path, then deleting all appearances of the edge $\{s,z\}$ and running the algorithm again to determine whether there is an equally good path with at least one internal vertex.  If we find that there is no pair $(s,z)$ such that some $\star$-optimal temporal $(s,z)$-path contains at least one internal vertex, it is immediate from the definition of temporal betweenness centrality that $\max_{v \in V}C_B^{(\star)}(v) = 0$, so in this case we return zero and terminate.

If we have not terminated at this point, we know that there exists some pair $(s,z) \in V^2$ such that $N_{sz} := \sum_{v \in V}\sigmageneric{sz}(v) \ge 1$.  Note that the number of $\star$-optimal temporal $(s,z)$-paths with no internal vertex is at most the number of appearances of the edge $\{s,z\}$ in the input graph, and thus is upper bounded by $T$.  It follows that $\sigmageneric{sz} \le N_{sz} + T$.  Summing the temporal betweenness centrality over all vertices, we see that
\[
    \sum_{v \in V} C_B^{(\star)}(v) \ge \frac{N_{sz}}{N_{sz} + T} \ge \frac{1}{1+T}.
\]
It therefore follows by the pigeonhole principle that there exists some vertex $v$ such that $C_B^{(\star)}(v) \ge \frac{1}{n(T+1)}$.  We shall exploit this fact later in the proof.

The remainder of our algorithm proceeds as follows; note that, by Lemma \ref{lma:count-to-sample}, there exists an FPAUS to sample  $\star$-optimal temporal $(s,z)$-paths in elements of $\mathcal{C}$, and we shall denote the output of a single run of this FPAUS on input graph $\mathcal{G}$ and vertices $s$ and $z$ with error parameter $\delta'$ by $\sample(\mathcal{G},s,z,\delta')$.  
\begin{enumerate}
    \item Initialise an $n$-element array $C$ with zeros.
    \item For each pair $(s,z) \in V^2$:
    \begin{itemize}
        \item If there exists at least one temporal $(s,z)$-path in $\mathcal{G}$, repeat $\ell := 300 000 \varepsilon^{-3} (T+1) n^3 \ln n$ times: 
            \begin{itemize}
                \item $P \leftarrow \sample(\mathcal{G},s,z,\varepsilon/20)$;
                \item For $1 \le i \le n$, if $v_i \notin \{s,z\}$ and $v_i$ belongs to $P$ then increment $C[i]$.
            \end{itemize}
    \end{itemize}
    \item Select $i \in [n]$ such that $C[i] = \max_{1 \le j \le n} C[j]$.  Return $C[i]/\ell$.
\end{enumerate}
It is clear that this procedure runs in time polynomial in $n$ and $\varepsilon^{-1}$ (note that, in the interests of a simple proof, we have made no attempt to optimise constants in the running time).  We claim that in fact, with probability at least $2/3$, $C[i]/\ell$ is an $\varepsilon$-approximation to $b^* := \max_{v \in V} C_B^{(\star)}(v)$. 

We begin by introducing some notation.  Given $s,z,v \in V$, we will write 
\[
b_{(s,z)}(v) = \begin{cases}
                \sigmageneric{sz}(v)/\sigmageneric{sz}    & \text{if $s \neq v \neq z$ and there exists at least one temporal $(s,z)$-path;} \\
                   0    & \text{otherwise};
                \end{cases}
\] 
thus $C_B^{(\star)}(v) = \sum_{(s,z) \in V^2} b_{(s,z)}(v)$.  For any vertex $v_i \in V$, let the random variable $X_i^{(s,z)}$ be the number of temporal $(s,z)$-paths we sample that contain $v_i$ as an internal vertex.  Set $X_i = \sum_{(s,z) \in V^2} X_i^{(s,z)}$ to be the total number of temporal paths we sample containing $v_i$ as an internal vertex.

We will consider the contribution to our path count $X_i$ from each of the random variables $X_i^{(s,z)}$. If none of these random variables has an expectation that is too small, we can argue that it is unlikely that any of them will take a value far from its expectation, so when we take the sum we will obtain a good approximation to $X_i$.  However, some of the random variables may have very low expectation, in which case they are likely to take values that differ from their expectations by a large multiplicative factor.  We treat such random variables separately, arguing that they are unlikely to take a large absolute value, and therefore that the total contribution from all such random variables is likely to have a negligible impact on the overall sum.

We define two events that correspond to the random variables behaving well in this way.  Let $\mathcal{F}_1$ be the event that, for every pair $(s,z) \in V^2$ and $i \in [n]$ with $\mathbb{E}(X_i^{(s,z)}) \ge \varepsilon \ell/ 60(T+1)n^3$, we have $|X_i^{(s,z)} - \mathbb{E}(X_i^{(s,z)})| \le (\varepsilon/20)\mathbb{E}(X_i^{(s,z)})$, and let $\mathcal{F}_2$ be the event that, for every pair $(s,z) \in V^2$ and $i \in [n]$ with $\mathbb{E}(X_i^{(s,z)}) < \varepsilon \ell/ 60(T+1)n^3$, we have $X_i^{(s,z)} \le \varepsilon \ell / 10 (T+1) n^3$.  Later, we will derive a lower bound on the probability that both $\mathcal{F}_1$ and $\mathcal{F}_2$ hold simultaneously.  First, we show that, conditioned on this event $\mathcal{F}_1 \cap \mathcal{F}_2$, the value returned by the algorithm will be an $\varepsilon$-approximation to $b^*$.

This argument proceeds in two stages.  First, we show that the vertex $v_i$ such that we return $C[i]$ at Step~3 of the algorithm satisfies $C_B^{(\star)}(v_i) \ge (1 - \varepsilon/2)b^*$.  Second, we show that $C[i]/\ell$ is an $(\varepsilon/2)$-approximation to $C_B^{(\star)}(v_i)$.  Together, these two facts imply that we do indeed return a $\varepsilon$-approximation to $b^*$, as required.  

For both stages, we will make use of bounds on $X_i$. Before deriving these bounds, we note that the probability $p_i^{(s,z)}$ that a single temporal $(s,z)$-path sampled using our FPAUS contains $v_i$ as an internal vertex satisfies $(1 - \varepsilon/20)b_{(s,z)}(v_i) \le p_i^{(s,z)} \le (1+\varepsilon/20)b_{(s,z)}(v_i)$.  Since each invocation of $\sample(\mathcal{G},s,z,\varepsilon/20)$ uses independent randomness, we see that $X_i^{(s,z)}$ has binomial distribution $Bin(\ell,p_i^{(s,z)})$.

Now, for an upper bound on $X_i$, note that (assuming $\mathcal{F}_1 \cap \mathcal{F}_2$ holds) we have
\begin{align}\label{eqn:X_i-upper}
    X_i &= \sum_{\substack{(s,v) \in V^2 \\ \mathbb{E}(X_i^{(s,z)}) \ge \varepsilon \ell/ 60(T+1)n^3}} X_i^{(s,z)} + \sum_{\substack{(s,v) \in V^2 \\ \mathbb{E}(X_i^{(s,z)}) < \varepsilon \ell/ 60(T+1)n^3}} X_i^{(s,z)} \nonumber \\
        &\le \sum_{(s,v) \in V^2} (1 + \varepsilon/20) \mathbb{E}(X_i^{(s,z)}) + \sum_{(s,v) \in V^2} \frac{\varepsilon \ell}{10 (T+1) n^3} \nonumber \\
        & \le \sum_{(s,v) \in V^2} (1 + \varepsilon/20)^2 \ell b_{(s,z)}(v_i) + \frac{\varepsilon \ell}{10 (T+1) n} \nonumber \\
        & \le (1 + 3 \varepsilon/20) \ell \sum_{(s,v) \in V^2} b_{(s,z)}(v_i) + \frac{\varepsilon \ell}{10 (T+1) n} \nonumber \\
        & = (1 + 3 \varepsilon/20) \ell C_B^{(\star)}(v_i) + \frac{\varepsilon \ell}{10 (T+1) n}.
\end{align}
On the other hand, for a lower bound, we observe that
\begin{align}\label{eqn:X_i-lower}
    X_i &\ge \sum_{\substack{(s,v) \in V^2 \\ \mathbb{E}(X_i^{(s,z)}) \ge \varepsilon \ell/ 60(T+1)n^3}} X_i^{(s,z)} \nonumber \\
        &\ge \sum_{\substack{(s,v) \in V^2 \\ \mathbb{E}(X_i^{(s,z)}) \ge \varepsilon \ell/ 60(T+1)n^3}} (1 - \varepsilon/20) \mathbb{E}(X_i^{(s,z)}) \nonumber \\
        &= \sum_{(s,z) \in V^2} (1 - \varepsilon/20) \mathbb{E}(X_i^{(s,z)}) \quad - \sum_{\substack{(s,v) \in V^2 \\ \mathbb{E}(X_i^{(s,z)}) < \varepsilon \ell/ 60(T+1)n^3}} (1 - \varepsilon/20) \mathbb{E}(X_i^{(s,z)}) \nonumber \\
        &\ge \sum_{(s,z) \in V^2} (1 - \varepsilon/20) \mathbb{E}(X_i^{(s,z)}) \quad - \sum_{(s,v) \in V^2} \frac{(1 - \varepsilon/20) \varepsilon \ell}{60(T+1)n^3} \nonumber \\
        & \ge \sum_{(s,v) \in V^2} (1 - \varepsilon/20)^2 \ell b_{(s,z)}(v_i) - \frac{\varepsilon \ell}{60 (T+1) n} \nonumber \\
        & \ge (1 - \varepsilon/10) \ell \sum_{(s,v) \in V^2} b_{(s,z)}(v_i) \nonumber - \frac{\varepsilon \ell}{60 (T+1) n} \\
        & = (1 - \varepsilon/10) \ell C_B^{(\star)}(v_i) - \frac{\varepsilon \ell}{60 (T+1) n}.
\end{align}

We now prove that, conditioned on the event $\mathcal{F}_1 \cap \mathcal{F}_2$, the vertex $v_i$ such that we return $C[i]/\ell$ at Step~3 of the algorithm satisfies $C_B^{(\star)}(v_i) \ge (1 - \varepsilon/2)b^*$.  To see this, suppose that we have $C_B^{(\star)}(v_i) < (1 - \varepsilon/2)b^*$, and fix a vertex $v_j$ such that $C_B^{(\star)}(v_j) = b^*$.  It suffices to demonstrate that $X_j > X_i$, as in this case we will not return a value corresponding to $v_i$.  Using the bounds in \eqref{eqn:X_i-upper} and \eqref{eqn:X_i-lower}, we see that
\begin{align*}
    X_j - X_i &\ge (1 - \varepsilon/10) \ell C_B^{(\star)}(v_j) - \frac{\varepsilon \ell}{60(T+1)n} - (1 + 3 \varepsilon/20) \ell C_B^{(\star)}(v_i) - \frac{\varepsilon \ell}{10 (T+1) n} \\
        &\ge (1 - \varepsilon/10) \ell b^* - (1 + 3 \varepsilon/20)(1 - \varepsilon/2)\ell b^* - \frac{7 \varepsilon \ell}{60 (T+1) n} \\
        &\ge (1 - \varepsilon/10) \ell b^* - (1 - 7 \varepsilon/20)\ell b^* - \frac{7 \varepsilon \ell}{60(T+1)n}\\
        &= \frac{\varepsilon}{4}\ell b^* - \frac{7 \varepsilon \ell}{60 (T+1) n}.
\end{align*}
Recall that we know $b^* \ge \frac{1}{(T+1)n}$; it follows that 
\[
    X_j - X_i \ge \varepsilon \ell b^* \left(\frac{1}{4} - \frac{7}{60} \right) > 0,
\]
as required.

We now proceed to the second stage of the argument, showing that, conditioned on $\mathcal{F}_1 \cap \mathcal{F}_2$, if $v_i$ is the vertex corresponding to the value $C[i]/\ell$ returned in Step 3 of the algorithm, $X_i/\ell$ is a $\varepsilon/2$-approximation to $C_B^{(\star)}(v_i)$.  From the bounds \eqref{eqn:X_i-upper} and \eqref{eqn:X_i-lower} on $X_i$, it follows immediately that
\[
    \left| X_i - \ell C_B^{(\star)}(v_i) \right| \le \frac{3 \varepsilon \ell C_B^{(\star)}(v_i)}{20} + \frac{\varepsilon \ell}{10(T+1)n}.
\]
Recall that, from the previous stage, we know that $C_B^{(\star)}(v_i) \ge (1 - \varepsilon/2)b^* \ge b^*/2$, using the assumption that $\varepsilon < 1$.  Since we know that $b^* \ge \frac{1}{(T+1)n}$, it follows that $C_B^{(\star)}(v_i) \ge \frac{1}{2(T+1)n}$, and so
\[
    \left| X_i - \ell C_B^{(\star)}(v_i) \right| \le \frac{3 \varepsilon \ell C_B^{(\star)}(v_i)}{20} + \frac{\varepsilon \ell C_B^{(\star)}(v_i)}{5} = \frac{7}{20} \varepsilon \ell C_B^{(\star)}(v_i) < \frac{1}{2} \varepsilon \ell C_B^{(\star)}(v_i).
\]
We therefore see that $|X_i/\ell - C_B^{(\star)}(v_i)| < (\varepsilon/2)C_B^{(\star)}(v_i)$, and hence $X_i/\ell$ is an $(\varepsilon/2)$-approximation to $C_B^{(\star)}(v_i)$, as required.

We have therefore demonstrated that our algorithm returns an $\varepsilon$-approximation to $b^*$ whenever the event $\mathcal{F}_1 \cap \mathcal{F}_2$ holds; it remains to prove that this event holds with probability at least $2/3$.  We shall in fact prove that each of $\mathcal{F}_1$ and $\mathcal{F}_2$ holds with probability at least $5/6$, so that the claimed result follows immediately by a union bound.  

We start by considering $\mathcal{F}_1$.  Let $Y$ be a binomial random variable whose expectation is at least $\varepsilon \ell / 60 (T+1) n^3$.  By Lemma~\ref{lem:chernoff}(i), we have
\begin{align*}
    \mathbb{P}\left(|Y - \mathbb{E}(Y)| \ge (\varepsilon/20) \mathbb{E}(Y) \right) &\le 2 \exp \left( - \frac{\varepsilon^2}{400} \frac{\mathbb{E}(Y)}{3} \right) \\
    &\le  2 \exp \left( - \frac{\varepsilon^3 \ell}{72000 (T+1) n^3}\right)\\
    &< 2 \exp \left( - 4 \ln n \right) \\
    &< 1/6n^3.
\end{align*}
Taking a union bound over all possible choices of $s$, $z$ and $i$, we therefore see that the probability that any random variable $X_i^{(s,z)}$ whose expectation is at least $\varepsilon \ell / 60(T+1)n^3$ takes a value that differs from its expectation by more than $(\varepsilon/20)\mathbb{E}(X_i^{(s,z)})$ is at most $1/6$.  It follows immediately that $\mathcal{F}_1$ holds with probability at least $5/6$.

Finally, we consider $\mathcal{F}_2$.  Let $Y$ be any binomial random variable whose expectation is at most $\varepsilon \ell/ 60 (T+1) n^3$.  By Lemma~\ref{lem:chernoff}(ii), we have
\[ \mathbb{P}\left(Y \ge \frac{\varepsilon \ell}{10 (T+1) n^3}\right) \le 2^{- \varepsilon \ell / 10 (T+1) n^3} = 2^{- 30 000 \varepsilon^{-2} \ln n} < 1/6n^3. \]
Again taking a union bound over all possible choices of $s$, $z$ and $i$, we therefore see that the probability that any random variable $X_i^{(s,z)}$ whose expectation is at most $\varepsilon \ell/ 60 (T+1) n^3$ takes a value greater than $\varepsilon \ell/ 10 (T+1) n^3$ is at most $1/6$, so $\mathcal{F}_2$ holds with probability at least $5/6$.  This completes the proof.
\end{proof}

Combining \cref{thm:between-approx} with \cref{cor:paths-short} gives the following immediate corollary.

\begin{corollary}
There is an FPTRAS which, given as input a temporal graph $\mathcal{G} = (V,\mathcal{E},T)$, computes an approximation to $\max_{v \in V} C_B^{(\star)}(v)$ (for $\star \in \{\text{fastest, foremost}\}$), parameterised by either the vertex cover number or treedepth of the underlying input graph.
\end{corollary}

\section{Conclusion}
 In this work, we initiate the systematic study of the parameterised and approximation complexity of \TemporalPaths. We present parameterised and approximation hardness results and complement them with several parameterised exact and approximation algorithms. 

In terms of improving our results, we conjecture that it is possible to prove \#W[1]-hardness instead of $\oplus$W[1]-hardness for \TemporalPaths parameterised by the feedback vertex number of the underlying graph. Furthermore, we leave open whether our parameterised approximation results for vertex cover number or treedepth of the underlying graph can be improved from a classification standpoint by obtaining exact algorithms, or whether we can also show parameterised hardness for those cases.

We leave open to what extent our results transfer to the problem of counting \emph{strict} temporal $(s,z)$-paths, where the labels on the time-edges have to be strictly increasing. We conjecture that most 
of our results hold for the strict case. In fact, we believe that the MSO formulation used to obtain fixed-parameter tractability for the treewidth of the underlying graph combined with the lifetime can be simplified: in the strict case, a first-order formula should suffice, which would lead to fixed-parameter tractability in terms of the lifetime on any class of nowhere-dense graphs~\cite{grohe18}, or for the combined parameter of cliquewidth and lifetime~\cite{courcelle01}.  

We conjecture that our polynomial-time algorithm for the case where the underlying graph is a forest (\cref{sec:ptimealgs}) can be extended to the case where the underlying graph is series-parallel~\cite{eppstein1992parallel}. Recall that a forest can be seen as a series-parallel graph where only series compositions are used. The idea would be to extend the dynamic program to parallel compositions, exploiting the observation that any temporal path can visit the terminal vertices of a series-parallel graph at most once.

Finally, we believe that our FPT-algorithms presented in \cref{sec:forestgeneralization} for the timed feedback vertex number and feedback edge number of the underlying graph, respectively, can be adapted to count other types of temporal $(s,z)$-paths for which counting is in general \#P-hard.
This leads us to believe that the algorithms can be modified to count restless temporal $(s,z)$-paths~\cite{CHMZ21} and possibly also to count delay-robust $(s,z)$-routes~\cite{FMNR22}, since both of these path types can be found in polynomial time when the underlying graph of the input temporal graph is a forest.

\bibliographystyle{abbrvnat}
\bibliography{betweenrefs}

\end{document}